\documentclass[1p,final]{elsarticle}
\usepackage{amsfonts,natbib,pslatex}
\usepackage{latexsym,amsmath,amsthm,amssymb,amsxtra,mathrsfs,times,CJK}
\usepackage{color}
\usepackage{amsfonts,color}
\usepackage{amssymb,amsmath,latexsym}
\usepackage{amssymb}
\usepackage{makecell}
\usepackage{amsfonts}
\usepackage{indentfirst}
\usepackage{amsmath}
\usepackage{amsthm}
\usepackage{amssymb}
\usepackage{color}
\usepackage{graphicx}
\usepackage[para]{threeparttable}
\usepackage{booktabs}
\usepackage{caption}
\usepackage{makecell}
\usepackage{diagbox}
\usepackage{multirow}
\usepackage{epstopdf}
\usepackage{subfigure}
\usepackage{booktabs}
\usepackage{hyperref}

\usepackage[para]{threeparttable}

\newcommand{\ord}{{\rm ord}}

\newcommand{\lcm}{{\rm lcm}}

\newtheorem{theorem}{Theorem}[section]
\newtheorem{lemma}[theorem]{Lemma}
\newtheorem{corollary}[theorem]{Corollary}

\newtheorem{example}[theorem]{Example}
\newtheorem{Proposition}[theorem]{Proposition}
\newtheorem{remark}[theorem]{Remark}

\begin{document}

	\begin{frontmatter}

		\title{MDS and AMDS symbol-pair codes constructed from  repeated-root codes \tnoteref{fn1}}

		\author[SWJTU]{Xiuxin Tang}
		\ead{XiuxinTang@my.swjtu.edu.cn}
		\author[SWJTU]{Rong Luo\corref{cor1}}
		\ead{luorong@swjtu.edu.cn}

		\cortext[cor1]{Corresponding author}
		\address[SWJTU]{School of Mathematics, Southwest Jiaotong University, Chengdu, 610031, China}

		\begin{abstract}
		Symbol-pair codes introduced by Cassuto and Blaum in 2010 are designed to protect against the pair errors in symbol-pair read channels. One of the central themes in symbol-error correction is the construction of maximal distance separable (MDS)  symbol-pair codes that possess the largest possible pair-error correcting performance.  Based on repeated-root cyclic codes, we construct  two classes of MDS symbol-pair codes for more general generator polynomials and also give a new class of almost MDS (AMDS) symbol-pair codes with the length $lp$. In addition, we derive all MDS and AMDS symbol-pair codes with length  $3p$, when the degree of the generator polynomials is no more than 10. The main results are obtained by determining the solutions of certain equations over finite fields.
		\end{abstract}
		
		
	\end{frontmatter}
	
	\section{Introduction}
	
Cassuto and Blaum (2010) proposed a new coding framework called symbol-pair codes to combat symbol-pair errors over symbol-pair read channels in \cite{Cassuto-2010-SIT} with the development of high-density data storage technologies. For example, Blu-ray disc is a high-density data storage symbol-pair read channels for the practical application. The seminal works \cite{Cassuto-2010-SIT}--\cite{Cassuto-2011-SIT} established relationships between an error-correcting code's minimum Hamming distance and the minimum symbol-pair distance, discovered methods for code construction and decoding, and obtained lower and upper bounds on code size. If a code $\mathcal{C}$ over ${\mathbb F}_p^{{n}}$ with length $n$ contains $M$ elements and has the minimum symbol-pair distance ${d_p}$, then $\mathcal{C}$ is referred as an ${\left( {n,M,{d_p}} \right)_p}$ symbol-pair code. Finding symbol-pair codes with high symbol-pair error correcting performance has become a significant theoretical challenge. 

In 2012, Chee et al. \cite{Chee-2012-SIT} established the Singleton-type bound on symbol-pair codes. Similar to the classical codes, the symbol-pair codes meeting the Singleton-type bound are called MDS symbol-pair codes and almost MDS symbl-pair codes are denoted AMDS symbol-pair codes. MDS symbol-pair codes are the most useful and interesting symbol-pair codes due to their optimality. Many researchers used various mathematical tools to try to obtain MDS symbol-pair codes. Constructing MDS symbol-pair codes is thus important both in theory and in practice. However, because determining the exact values of symbol-pair distances of constacyclic codes is a very complicated and difficult task in general, little work has been done on it.

In 2013, Chee et al. \cite{Chee-2013-SIT} obtained some MDS symbol-pair codes from classical MDS codes.
Between 2015 and 2018, researchers  \cite{Kai-2015-TIT}\cite{Kai-2018-NCT}\cite{Li-2017-DCC} constructed some MDS symbol-pair codes with minimum symbol-pair distance 5 and 6 from constacyclic codes.
In 2017, Chen et al. \cite{Chen-2017-CIT} proposed to construct MDS symbol-pair codes by repeated-cyclic codes and obtained some length $3p$ MDS symbol-pair codes with symbol-pair distance 6 to 8 and MDS ${\left( {lp,5} \right)_p}$  symbol-pair codes. In the next few years, MDS symbol-pair codes with some new parameters were found  by using repeated-root codes over ${\mathbb F}_p$. 
In 2018, Dinh et al. \cite{Dinh-2018-TIT} presented all MDS symbol-pair codes with prime power lengths by repeated-root constacyclic codes. Two years later, Dinh et al. \cite{Dinh-2020-TIT} also constructed a families of MDS sympol-pair codes with length $2p^s$. 
In the next 2019 to 2021, Zhao \cite{Zhao-2019-Doc} \cite{Zhao-2021-CL} 
constructed some MDS symbol-pair codes using repeated-root constacyclic codes.
In 2022, Ma et al. \cite{Ma-2022-DCC} obtained MDS $ (3p,10)_{p} $  and MDS $ (3p,12)_{p} $ sympol-pair codes.

Inspired by these works, in order to obtain longer and more flexible code length, as well as a larger minimum symbol-pair distance, this paper proves that there are more general generator polynomials for MDS $ {(lp,6)}_p $ and MDS $ {(lp,5)}_p $ symbol-pair codes by using repeated-cyclic codes. Furthermore, the parameter AMDS $ {(lp,7)}_p $  symbol-pair codes are obtained by using repeated-cyclic codes. For length $n=3p$, this paper gives all MDS and AMDS symbol-pair codes from repeated-root cyclic codes $\mathcal{C}_{(r_{1},r_{2},r_{3})}$, when the degree of the generator polynomial $g_{(r_{1},r_{2},r_{3})}(x)$ is no more than 10, i.e., $\deg( g_{(r_{1},r_{2},r_{3})}(x))\le 10$.  

The rest of this paper is structured as follows. In Section 2, we introduce some basic notations and results on symbol-pair codes. In Section 3, we derive some new classes of MDS symbol-pair codes and AMDS symbol-pair codes from repeated-root cyclic codes. In section 4, we conclude the paper.

\section{Preliminaries}

In this section, we introduce some notations and auxiliary tools on symbol-pair codes, which will be used to prove our main results in the sequel. We donote that ${\mathbb F}_p$ and ${\mathbb F}_q$ are finite fields, where $p$ is an odd prime and $q=p^m$. Then we denote ${\mathbb F}_p^{\text{*}}$ is the cyclic group ${\mathbb F}_p\verb|\|\{0\}$.
Let $${\bf{\textit{\textbf{x}}}} = \left( {{x_0},{x_1}, \cdots ,{x_{n - 1}}} \right)$$ be a vector in ${\mathbb F}_p^{n}$. Then the symbol-pair read vector of $\textit{\textbf{x}}$ is
$$\pi \left( {\bf{\textit{\textbf{x}}}} \right) = \left[ {\left( {{x_0},{x_1}} \right),\left( {{x_1},{x_2}} \right), \cdots ,\left( {{x_{n - 1}},{x_0}} \right)} \right].$$
Similar to the Hamming weight ${\omega _H}\left( {\bf{\textit{\textbf{x}}}} \right)$ and Hamming distance $ {D_H}\left( \bf{\textit{\textbf{x}}},\bf{\textit{\textbf{y}}} \right) $. The symbol-pair weight ${\omega _p}\left( {\bf{\textit{\textbf{x}}}} \right)$ of the symbol-pair vector ${\bf{\textit{\textbf{x}}}}$ is defined as
$${\omega _p}\left( {\bf{\textbf{\textit{x}}}} \right) = \left| {\left\{ {i\left| {\left( {{x_i},{x_{i + 1}}} \right) \ne \left( {0,0} \right)} \right.} \right\}} \right|.$$
The symbol-pair distance $ {D_p}\left( \bf{\textit{\textbf{x}}},\bf{\textit{\textbf{y}}} \right) $ between any two vectors $\bf{\textit{\textbf{x}}},\bf{\textit{\textbf{y}}}$ is
$${D_p}\left( \bf{\textit{\textbf{x}}},\bf{\textit{\textbf{y}}} \right) = \left| {\left\{ {\left. i \right|\left( {{x_i},{x_{i + 1}}} \right) \ne \left( {{y_i},{y_{i + 1}}} \right)} \right\}} \right|.$$
The minimum symbol-pair distance of a code $\mathcal{C}$ is
$${d_p} = \min \left\{ { {{D_p}\left( \bf{\textit{\textbf{x}}},\bf{\textit{\textbf{y}}} \right)\left|\; \bf{\textit{\textbf{x}}},\bf{\textit{\textbf{y}}} \in \mathcal{C} \right.} } \right\},$$
and we denote $ (n,k,d_{p})_{p} $ a symbol-pair code   with length $ n $, dimension $ k $ and minimum symbol-pair distance $ d_{p} $ over  ${\mathbb F}_p$. For any code $\mathcal{C}$ of length $n$ with $0 < d_{H}(\mathcal{C}) < n$, there is an important inequality  between $d_{H}(\mathcal{C})$ and $d_{p}(\mathcal{C})$ 
in \cite{Cassuto-2011-TIT}: $$d_{H}(\mathcal{C})+1 \le  d_{p}(\mathcal{C}) \le  2d_{H}(C).$$

In this paper, we always regard the codeword \textbf{c} in $\mathcal{C}$ as the corresponding polynomial $c(x)$.
The following lemmas will be applied in our later proofs.

Similar to classical error-correcting codes, the size of symbol-pair codes satisfies the following Singleton bound. The symbol-pair code achieving the Singleton bound is called a maximum distance separable (MDS) symbol-pair code.

\begin{lemma}\label{lemma 3.4 }	
	
	(\cite{Chee-2012-SIT})If $\mathcal{C}$ is a symbol-pair code with length $n$ and minimum symbol-pair distance ${d_p}$ over ${\mathbb F}_q$, we call an  $ (n,d_p)_{p} $ symbol-pair code of size $ q^{n - {d_p} + 2} $ maximum distance separable (MDS)  and an  $ (n,d_p)_{p} $ symbol-pair code of size $ q^{n - {d_p} + 1} $ almost maximum distance separable (AMDS) for $ q\geqslant 2 $.
	
\end{lemma}

\begin{lemma}\label{lemma 3.2 }
	(\cite{Chee-2013-SIT}) Let $q \geqslant 2$ and $2 \leqslant {d_p} \leqslant n$. If $\mathcal{C}$ is a symbol-pair code with length $n$ and minimum symbol-pair distance ${d_p}$ over ${\mathbb F}_q$, then $\left|\mathcal{C}\right| \leqslant {q^{n - {d_p} + 2}}.$
\end{lemma}

In some cases, the bound of minimum symbol-pair distance can be improved. 
\begin{lemma}\label{lemma 3.3 }
	(\cite{ Chen-2017-CIT}) Let $\mathcal{C}$ be an $\left[ {n,k,{d_H}} \right]$ constacyclic code over ${{\mathbb F}_q}$ with ${\text{2}} \leqslant {{\text{d}}_{\text{H}}}\leqslant{\text{n}}$. Then we have the following ${d_p}\left( {\mathcal{C}} \right) \geqslant {d_H} + 2$ if and only if $\mathcal{C}$ is not an MDS code.
\end{lemma}

The next lemma will be used by the later proof of Proposition \ref{Proposition C_{421}}. 

\begin{lemma}\label{lemma 3.6 }
	(\cite{Chen-2017-CIT}) Let $n=3p$ with $p \equiv 1\left( {\bmod~3} \right)$. If $\mathcal{C}_{(3,2,1)}$ is the cyclic code in ${{\mathbb F}_p}\left[ x \right]/\left\langle {{x^n} - 1} \right\rangle $ generated by 
	\[g_{(3,2,1)}(x) = {(x - 1)^3}{(x - \omega )^2}(x - {\omega ^2}),\]
	then $\mathcal{C}_{(3,2,1)}$ is an MDS $ (3p,8)_{p} $ symbol-pair code, where $\omega$ is a primitive $3$-th  root of unity in ${{\mathbb F}_p}$.
	
\end{lemma}

Researchers constructed some MDS symbol-pair codes with minimum symbol-pair distance 6 by repeated-root cyclic codes. The next lemma will be used by several parts of Theorem \ref{theorem 3.1}.

\begin{lemma} \label{lemma 3.7 }
	\item 
	\begin{itemize}
		\item \cite{Dinh-2018-TIT} $ \mathcal{C}= \left\langle {(x - 1)}^4 \right\rangle $ is an MDS symbol-pair code.	
		\item \cite{Dinh-2020-TIT} ${\mathcal{C}} = \left\langle {{{(x - 1)}^i}{{(x + 1 )}^j}} \right\rangle $ is an MDS symbol-pair code with $d_p=6$ over ${{\mathbb F}_p}$, where ${\left| {i - j} \right| \leqslant 2}$ and $i,j\le p-1$.
		\item \cite{Zhao-2019-Doc} ${\mathcal{C}} = \left\langle {{{(x - 1)}^3}{{(x - \omega )}}} \right\rangle $ is an MDS ${\left( {lp,{\text{\;}}6} \right)_p}$ symbol-pair code over ${{\mathbb F}_p}$, where $\omega$ is a primitive $ l$-th  root of unity in ${{\mathbb F}_p}$.
	\end{itemize}
		
\end{lemma}

The method for calculating the minimum Hamming distance about repeated-cyclic codes is given in the following lemma.
\begin{lemma}\label{lemma 3.1}
	(\cite{Castagnoli-1991-OIT}) Let $\mathcal{C}$ be a repeated-root cyclic code with length $ lp^{e} $ over ${{\mathbb F}_q}$ generated by $ g(x) = \prod {{m_i}^{{e_i}}\left( x \right)}$ , where $ l $ and $ e $ are positive integers with $ \gcd (l, p) = 1 $. Then we have $${d_H}(\mathcal{C}) = \min \left\{ {{P_t} \cdot {d_H}\left( {{{\mathcal{\overline C}}_t}} \right)\left| {0 \le t \le l{p^e}} \right.} \right\},$$
	where $ {P_t} = {\omega _H}\left( {{{\left( {x - 1} \right)}^t}} \right) $ and $ \mathcal{\overline C}_{t} =\left\langle \prod\limits_{{e_i} > t} {{m_i}\left( x \right)}\right\rangle $. 
\end{lemma}

\section{Constructions of MDS and AMDS Symbol-Pair Codes}
In this section, we propose some MDS and AMDS symbol-pair codes by repeated-root cyclic codes over  ${{\mathbb F}_p}$, where $ p $ is an odd prime and ${{\mathbb F}_p}$ is a $p$-ary finite field.

\subsection{MDS and AMDS Symbol-Pair Codes with length $lp$}
In this subsection, we prove that there exist more general generator polynomials about MDS $ {(lp,6)}_p $ and MDS $ {(lp,5)}_p $ symbol-pair codes. Furthermore, the parameter AMDS $ {(lp,7)}_p $  symbol-pair codes are obtained by using repeated-cyclic codes.

Let $\mathcal{C}_{(r_{1},r_{2},r_{3})}$ be the repeated-root cyclic code over ${{\mathbb F}_p}$ and the generator ploynomial of $ \mathcal{C}_{(r_{1},r_{2},r_{3})} $  is
$$g_{(r_{1},r_{2},r_{3})} \left( x \right) = {\left( {x - 1} \right)^{r_1}}\left( {{x} + 1} \right)^{r_2}\left( {{x} - \omega} \right)^{r_3},$$
where $\omega$ is a primitive $ l$-th  root of unity in ${{\mathbb F}_p}$ and ${r_1}+{r_2}+{r_3}=4$. 
Dinh \cite{Dinh-2018-TIT} \cite{Dinh-2020-TIT} discussed all cases of $l=1$ and $l=2$, here we focus on the case of $l > 2$.

\begin{theorem} \label{theorem 3.1}
Let	$\mathcal{C}_{(r_{1},r_{2},r_{3})}$ be an MDS symbol-pair code with $d_p=6$, if ${r_1}$, ${r_2}$ and ${r_3}$ meet the conditions in Table \ref{table-1}.
\end{theorem}

\begin{table}[h]\label{table1}
	\centering
	\begin{threeparttable}[b]
		\caption{MDS symbol-pair codes of Theorem \ref{theorem 3.1}}\label{table-1}
		\begin{tabular}{@{}lllll@{}}			
			\toprule
			
			$r_1$ & $r_2$ & $r_3$ & \makecell{$(n,d_{p})_{p}$} & Reference or Proposition\\
			\midrule	
			4 & 0 & 0 & \makecell{$(p,6)_{p}$}&Reference \cite{Dinh-2018-TIT}\\
			\midrule 
			3 & 0 & 1 & \makecell{$(lp,6)_{p}$}&Reference \cite{Zhao-2019-Doc} \\
			\midrule
			2 & 1 & 1 & \makecell{$(klp,6)_{p}$}&Proposition \ref{Proposition 3.2}\\ 					
			\midrule
			2 & 0 & 2 & \makecell{$(lp,6)_{p}$ } &Proposition \ref{Proposition 3.4}\\
			\midrule
			
		\end{tabular}	
		\begin{tablenotes}
			\item * {\footnotesize When $l$ even, $(klp,6)_{p}=(lp,6)_{p}$; When $l$ odd, $(klp,6)_{p}=(2lp,6)_{p}$.}
			
			\item * {\footnotesize Let $l$ be an odd number or $l \equiv 0\left( {\bmod~4} \right)$, if $(r_{1},r_{1},r_{1}) \in S$ and $S=\{(0,1,3),(0,2,2),(0,3,1)\}$.}
		\end{tablenotes}
	\end{threeparttable}
\end{table}

The proof of Theorem \ref{theorem 3.1} needs Propositions \ref{Proposition 3.2} to Proposition \ref{Proposition 3.4}. For the case of generator polynomials with three factors ${\left( {x - 1} \right)},\left( {{x} + 1} \right)$ and $\left( {{x} - \omega} \right)$, we have the following proposition.

\begin{Proposition}\label{Proposition 3.2}
Let $ \mathcal{C}_{(r_{1},r_{2},r_{3})} $ be an MDS symbol-pair code with $d_p=6$, if $\left( {r_1},{r_2},{r_3}\right) \in S_{1}$ and $S_{1}=\{\left(2,1,1 \right), \left(1,2,1 \right), \left(1,1,2 \right)\} $.

\end{Proposition}
\begin{proof}
When	$\left( {r_1},{r_2},{r_3}\right) =\left(2,1,1 \right) $, let $\mathcal{C}_{(2,1,1)}$ be the cyclic code over ${{\mathbb F}_p}$ and generated by $$g_{(2,1,1)}\left( x \right) = {\left( {x - 1} \right)^2}\left( {{x} + 1} \right)\left( {{x} - \omega} \right).$$

\noindent	By   Lemma \ref{lemma 3.1} , let ${\overline g}_{t}(x) $ be the generator ploynomials of $ \mathcal{\overline C}_{t} $.

\begin{itemize}
	
	\item \noindent	If $ t=0 $, then we have $$ {\overline g}_{0}(x)=(x-1)(x+1)(x-\omega) .$$ 
	It is easy to verify that the minimum Hamming diatance is $3$ in $ \mathcal{\overline C}_{0} $ and $ P_{0}=1 $. Therefore, this indicates $ P_{0} \cdot d_{H}(\mathcal{\overline C}_{0})=3.$
	
	\item \noindent If $ t=1 $, then we have ${\overline g}_{1}(x)=(x-1) $ and $ P_{1}=2 $. Thus, one can derive that ${P_1} \cdot {d_H}({{\mathcal{\overline C}}_1}) = 4.$
	
	\item \noindent If $2 \le t \le p-1$, then we have ${\overline g}_{t}(x)=1 $ and $ P_{t} \ge 2 $. This implies that ${P_t} \cdot {d_H}({{\mathcal{\overline C}}_t}) \ge 3.$
	
	Therefore, it can be verified that $\mathcal{C}$ is an $\left[ {lp,\;lp - 4,\;3} \right]$ repeated-root cyclic code over ${{\mathbb F}_p}$. Lemma \ref{lemma 3.3 } yields that ${d_p} \geqslant 5$, since $\mathcal{C}$ is not an MDS cyclic code.
	
\end{itemize}

	
	If $c \in \mathcal{ C}_{(2,1,1)}$ has $ \omega_p =5$ with  $\omega_H =4 $,
	then its certain cyclic shift must have the form
	$$ \left( { \star ,\; \star ,\; \star ,\; \star ,{0_s}} \right),$$
	where each $ \star $ denotes an element in  ${\mathbb F}_p^{\text{*}}$ and ${0_s}$ is all-zero vector with length $s$. Without loss of generality, suppose that the constant term of $c\left( x \right)$ is 1. We denote that
	$$c\left( x \right) = 1 + {a_1}x + {a_2}{x^2} + {a_3}{x^3} .$$
	This is a contradiction for  $\deg\left( {c\left( x \right)} \right)  \ge   \deg\left( {g_{(2,1,1)}\left( x \right)} \right)$.
	
	 If $c \in \mathcal{ C}_{(2,1,1)}$ has  $ \omega_p =5$ with $\omega_H =3 $, then its certain cyclic shift must have the form
	$$\left( { \star ,\; \star ,\;{0_{s_1}},\; \star ,\;{0_{s_2}}} \right),$$
	where each $ \star $ denotes an element in ${\mathbb F}_p^{\text{*}}$ and ${0_{s_1}}$, ${0_{s_2}}$ are all-zero vectors with lengths $s_1$ and $s_2$ respectively. Without loss of generality, suppose that the constant term of $c\left( x \right)$ is 1. We denote that
	$$c\left( x \right) = 1 + {a_1}x + {a_2}{x^t}.$$
	
	When $ t $ is even, it can be verified that
	\begin{equation*}  
		\left\{  
		\begin{array}{lr}  
			1 + {a_1} + {a_2}  = 0, &  \\  
			1 - {a_1} + {a_2}= 0, &  
		\end{array}  
		\right.  
	\end{equation*}
	since $c\left( 1 \right) = c\left( { - 1} \right)  = 0$. 
	This is impossible, since $ {a_1} \ne 0 $ and $p\ne 2$.
	
	Similarly, if $ t $ is odd, one can obtain that
		\begin{equation*}  
		\left\{  
		\begin{array}{lr}  
			1 + {a_1} + {a_2}  = 0, &  \\  
			1 - {a_1} - {a_2}= 0, &  
		\end{array}  
		\right.  
	\end{equation*}
	 which contradicts $ p $ odd.


	Let $y = -x,z = \frac{x}{\omega }$, we can deduce the following results by deforming it,
\[\begin{gathered}
	{g_{(2,1,1)}}(x) = (x - 1)^{2}(x + 1)(x - {\omega}) \hfill \\
\;\;\;\;\;\;\;\;\;\;\;\;\;\;\;\;\;\;\,=(y + 1)^{2}(y - 1)(y + {\omega}) \hfill \\
\;\;\;\;\;\;\;\;\;\;\;\;\;\;\;\;\;\;\,= {g_{(1,2,1)}}(y) \hfill \\
\;\;\;\;\;\;\;\;\;\;\;\;\;\;\;\;\;\;\,= \omega ^4 {(z - \frac{1}{{{\omega }}})^2}{(z + \frac{\omega }{{{\omega }}})}{(z - \frac{{{\omega }}}{{{\omega }}})} \hfill \\
\;\;\;\;\;\;\;\;\;\;\;\;\;\;\;\;\;\;\,=\omega^4{g_{(1,1,2)}}(z).\hfill \\
\end{gathered} \]
Thus, we have $\mathcal{C}_{(2,1,1)}=\mathcal{C}_{(1,2,1)}=\mathcal{C}_{(1,1,2)}$.

\end{proof}

We have the following two propositions, when the generator polynomials only have two of these three factors  ${\left( {x - 1} \right)},\left( {{x} + 1} \right)$ and $\left( {{x} - \omega} \right)$.

\begin{Proposition} \label{Proposition 3.3}

Let $ \mathcal{C}_{(r_{1},r_{2},r_{3})} $ be an MDS symbol-pair code with $d_p=6$, if  $\left( {r_1},{r_2},{r_3}\right) \in S_{2}$ and $ S_{2}=\{\left(1,0,3 \right) ,\left(0,3,1 \right), \left(0,1,3 \right)\}$.
\end{Proposition}
\begin{proof}
When $\left( {r_1},{r_2},{r_3}\right) =\left(1,0,3 \right) $, let $\mathcal{C}_{\left(1,0,3 \right)}$ be a repeated-root cyclic code over ${{\mathbb F}_p}$ generated by $$g_{\left(1,0,3 \right)}\left( x \right) = {\left( {x - 1} \right)}\left( {{x} - \omega} \right)^3.$$

	\noindent By Lemma \ref{lemma 3.1},  we can derive $d_H=3$ and Lemma \ref{lemma 3.3 } implies that $d_p\ge 5$.
With arguments similar to the proof of Proposition \ref{Proposition 3.2}, there are no codewords of $\mathcal{C}_{\left(1,0,3 \right)}$ with  $\omega_H=4$ such that the 4 nonzero terms appear with consecutive coordinates. Next we show that there are no codewords of $\mathcal{C}_{\left(1,0,3 \right)}$ with $\omega_H=3$ in the form
	$$\left( { \star ,\; \star ,\;{0_{s_1}},\; \star ,\;{0_{s_2}}} \right),$$
	where each $ \star $ denotes an element in ${\mathbb F}_p^{\text{*}}$ and ${0_{s_1}}$, ${0_{s_2}}$ are all-zero vectors with lengths $s_1$ and $s_2$ respectively. Without loss of generality, suppose that the constant term of $c\left( x \right)$ is 1. We denote that
	$$c\left( x \right) = 1 + {a_1}x + {a_2}{x^t}.$$
	Then $c^{(1)}\left( \omega \right) = c^{(2)}\left( \omega \right) =0$ induces that $t-1=kp$ for some positive integers $k\le l-2$, together with $c\left( { \omega} \right) =0$, one can immediately get
	\begin{equation*}  
		\left\{  
		\begin{array}{lr}  
			1 + {a_1}\omega + {a_2}\omega^{k+1}  = 0, &  \\  
			{a_1} + {a_2}\omega^k= 0. &  
		\end{array}  
		\right.  
	\end{equation*}
	By solving the system, we can derive a contradiction, since $p$ is an odd prime.

For the case of $\left( {r_1},{r_2},{r_3}\right) \in S_{2}$, with similar to the proof of Proposition \ref{Proposition 3.2}, we can deduce $\mathcal{C}_{\left(1,0,3 \right)}=\mathcal{C}_{\left(0,3,1 \right)}=\mathcal{C}_{\left(0,1,3 \right)}$.

\end{proof}

\begin{Proposition} \label{Proposition 3.4}
$ \mathcal{C}_{\left( {r_1},{r_2},{r_3}\right)} $ is an MDS symbol-pair code with $d_p=6$, when $\left( {r_1},{r_2},{r_3}\right) \in S_{3}$ and $S_{3} =\{\left(2,0,2 \right), \left(0,2,2 \right) \}$.		
\end{Proposition}
\begin{proof}

When 	$\left( {r_1},{r_2},{r_3}\right) =\left(2,0,2 \right) $, let $\mathcal{C}_{\left(2,0,2 \right)}$ be a cyclic code over ${{\mathbb F}_p}$ generated by $$g_{\left(2,0,2 \right)}\left( x \right) = {\left( {x - 1} \right)^2}\left( {{x} - \omega} \right)^2.$$ 

\noindent By Lemma \ref{lemma 3.1}, we can derive that the minimum Hamming distance of $ \mathcal{C}_{\left(2,0,2 \right)} $ is $d_H=3$ and Lemma \ref{lemma 3.3 } implies $d_p\ge 5$. With arguments similar to the proof of Proposition \ref{Proposition 3.2}, there are no codewords of $\mathcal{C}_{\left(2,0,2 \right)}$ with $\omega_H=4$ such that the 4 nonzero terms appear with consecutive coordinates. 
Next, we show that there are no codewords of $\mathcal{C}_{\left(2,0,2 \right)}$ with $\omega_H=3$ in the form
	$$ \left( { \star ,\; \star ,\;{0_{s_1}},\; \star ,\;{0_{s_2}}} \right),$$
	where each $ \star $ denotes an element in ${\mathbb F}_p^{\text{*}}$ and ${0_{s_1}}$, ${0_{s_2}}$ are all-zero vectors  with lengths $s_1$ and $s_2$ respectively. Without loss of generality, suppose that the constant term of $c\left( x \right)$ is 1. We denote that
	$$c\left( x \right) = 1 + {a_1}x + {a_2}{x^t}.$$
	However, by $c^{(1)}\left( 1 \right)= c^{(1)}\left( \omega \right) = 0$, we can deduce that $l$ is a divisor of $t-1$. Then combined with $c\left( 1 \right)= c\left( \omega \right) = 0$, we have 
	\begin{equation*}  
		\left\{  
		\begin{array}{lr}  
			1 + {a_1} + {a_2}  = 0, &  \\  
			1 + {a_1}\omega + {a_2}\omega= 0. &  
		\end{array}  
		\right.  
	\end{equation*}	
   This implies $\omega=1$, which is impossible, since $\omega^l=1$ and $l>2$.


When $\left( {r_1},{r_2},{r_3}\right) =\left(0,2,2 \right)$, with similar to the proof of Proposition \ref{Proposition 3.2}, we can deduce $\mathcal{C}_{\left(2,0,2 \right)}=\mathcal{C}_{\left(0,2,2 \right)}$.

\end{proof}

This completes the proof of Theorem \ref{theorem 3.1} from Proposition \ref{Proposition 3.2} to Proposition \ref{Proposition 3.4}, Remark \ref{remark 3.5} to Remark \ref{remark dH=2} and Lemma \ref{lemma 3.7 }. Therefore, we find all MDS symbol-pair codes containing these three factors ${\left( {x - 1} \right)}$, $\left( {{x} + 1} \right)$ and $\left( {{x} - \omega} \right)$ with minimum symbol-pair distance 6.
In fact, Lemma \ref{lemma 3.7 } is a special form of Proposition \ref{Proposition 3.3} and Proposition \ref{Proposition 3.4} for $d_{p}=6$, where $\omega$ is a $p-1\over 2$-th primitive element in ${{\mathbb F}_p}$.

\begin{remark} \label{remark 3.5}
 When $l$ even, factor ${\left( {x + 1} \right)}$ is a divisor of $x^{l}-1$; when $l$ odd, factor ${\left( {x + 1} \right)}$ is a divisor of $x^{2l}-1$.
\end{remark}	
\begin{remark} \label{remark dH=2}		
\noindent If $l \equiv 2\left( {\bmod~4} \right)$ and $(r_{1},r_{1},r_{1}) \in S$ for $S=\{(0,1,3),(0,2,2),(0,3,1)\}$, the minimum Hamming distance of $\mathcal{C}_{(r_{1},r_{2},r_{3})}$ is 2.
\end{remark}

In what follows, we obtain more general generator polynomials for symbol-pair codes with length $n = lp$  and minimum symbol-pair distance 5 or 6.

If $m_{1}$ and $m_2$ are two positive integers, then $\lcm [m_1,m_2] $ is the lowest common multiple of $m_1$ and $m_2$, as well as $\gcd(m_1,m_2) $ is the greatest common divisor of $m_1$ and $m_2$.
Let $\mathcal{C}_{a}$ and $ \mathcal{C}_{b} $ be the repeated-root cyclic codes over ${{\mathbb F}_p}$ and the generator ploynomial of $ \mathcal{C}_{a} $ and $ \mathcal{C}_{b} $ are
$$g_{a}\left( x \right) = \left( {{x} - {\omega_{0}}^{t_1}} \right)^{r_1}\left( {{x} - {\omega_{0}}^{t_2}} \right)^{r_2}$$ 
and 
$$g_{b}\left( x \right) = \left( x - 1 \right)^{2}\left( x - {\omega_{0}}^{t_1} \right)\left(x - {\omega_{0}}^{t_2}\right)$$ 
respectively,  
where ${t_1}\ge{t_2},$ $\ord({\omega_{0}}^{t_1})={m_1}$, $\ord({\omega_{0}}^{t_2})={m_2}$, $\lcm[{m_1},{m_2}]=l,\;\gcd\left({t_1}-{t_2},l \right) =1,\;3\le {r_1}+{r_2}\le 4$ and $\omega_{0}$ is the primitive element in ${{\mathbb F}_p}$.

\begin{corollary}\label{theorem 3.6}
Let $\mathcal{C}_{a}$ be an MDS symbol-pair code,  if  ${r_1}\ne 0$ and ${r_2}\ne 0$.
\end{corollary}
\begin{proof}
There are three cases that need to be discussed, $({r_1},{r_2} ) =\left(2,1 \right) $,  $\left(1,3 \right) $ and $\left(2,2 \right)$. When $({r_1},{r_2} ) =\left(1,2 \right) $ and  $\left(3,1 \right) $ is satisfied, it is similar to $({r_1},{r_2} ) =\left(2,1 \right) $ and  $\left(1,3 \right) $. 

	\item [\textbf{Case I.}] For the case of $({r_1},{r_2} ) =\left(2,1 \right) $,
	if there exsits a nonzero codeword with $\omega_H=2$ in $\mathcal{C}_{a}$, 
	without loss of generality, suppose that the constant term of $c\left( x \right)$ is 1. We denote that
	$$c(x)=1+a_1x^t,$$ 
	where $a_1\ne 0$ and $t\ne 0$.	It follows from $c({\omega_{0}}^{t_1})=c({\omega_{0}}^{t_2})=0$ that 
	\begin{equation*}  
		\left\{  
		\begin{array}{lr}  
			1+{a_1} {\omega_{0}}^{tt_1} = 0, &  \\  
			1+{a_1} {\omega_{0}}^{tt_2}= 0. &    
		\end{array}  
		\right.
	\end{equation*}
	By solving the system, we have	${{\omega_{0}}^{t\left({{t_1}-{t_2}}\right)}}=1$. Together with $c^{(1)}({\omega_{0}}^{t_1})=0$ and $\gcd\left({t_1}-{t_2},l \right) =1$, one can immediately get $lp$ is a divisor of $t$, which contradicts with the code length $lp$. 
	
	Thus, combined with Lemma \ref{lemma 3.1}, the minimum Hamming distance of $\mathcal{C}_{a}$ is $d_H=3$. 
	By Lemma \ref{lemma 3.3 }, we have $\mathcal{C}_{a}$ is an MDS $ (lp,5)_{p} $ symbol-pair code.
	
	\item	[\textbf{Case II.}] For the case of	 $({t_1},{t_2} ) =\left(3,1 \right) $, we have the generator polynomial	
	$$g_{a}\left( x \right) = \left( {{x} - {\omega_{0}}^{t_1}} \right)^3\left( {{x} - {\omega_{0}}^{t_2}} \right).$$
	By the proof of $   {\bf{Case\; I}}$, since $ {\bf{Case\; II}}$ is a subcode of $ {\bf{Case\; I}}$, we can draw the conclusion that the minimum Hamming distance of $\mathcal{C}_{a}$ is $3$ in $  {\bf{Case\; II}}$, when  $({t_1},{t_2} ) =\left(3,1 \right) $. 
	
	If there  is a codeword with Hamming weight $3$ and symbol-pair weight $5$. Then its certain cyclic shift must be the following form
	$$ \left( { \star ,\; \star ,\;{0_{s_1}},\; \star ,\;{0_{s_2}}} \right),$$ where each $ \star $ denotes an element in ${\mathbb F}_p^{\text{*}}$ and ${0_{s_1}}$, ${0_{s_2}}$ are all-zero vectors  with lengths $s_1$ and $s_2$ respectively. Then we have a codeword polynomial
	$$c\left( x \right) = 1 + {a_1}x + {a_2}{x^t}.$$
	However, it follows from $c^{(1)}\left( {\omega_{0}}^{t_1} \right)= c^{(2)}\left( {\omega_{0}}^{t_1} \right) =0$ that 
	\begin{equation*}  
		\left\{  
		\begin{array}{lr}  
			{a_1} + t{a_2}{\omega_{0}}^{(t-1)t_1} = 0, &\\
			t(t-1){a_2}{\omega_{0}}^{(t-2)t_1} = 0. &    
		\end{array}  
		\right.  
	\end{equation*}
	By solving the system, we have ${p\,\left| \,{t - 1} \right.}$.	Then $c^{(1)}\left( {\omega_{0}}^{t_1} \right)= c\left( {\omega_{0}}^{t_1} \right) =0$  indicates
	\begin{equation*}  
		\left\{  
		\begin{array}{lr}  
			1+	{a_1}{\omega_{0}}^{t_1} + {a_2}{\omega_{0}}^{{t_1}t} = 0, & \\ 
			{a_1}  + {a_2}{\omega_{0}}^{{t_1}(t-1)} = 0. &   
		\end{array}  
		\right.  
	\end{equation*}	
Then, we can derive a contradiction, since $p$ is an odd prime.	
	
	Therefore, there does not exsit a nonzero codeword with Hamming weight $3$ and symbol-pair weight $5$. Then, $\mathcal{C}_{a}$ is an MDS $ (lp,6)_{p} $ symbol-pair code, when $({r_1},{r_2} ) =\left(3,1 \right) $.
	
	\item [\textbf{Case III.}] For the case of	 $({t_1},{t_2} ) =\left(2,2 \right) $, similarly, we have the generator polynomial	
	$$g_{a}\left( x \right) = \left( {{x} - {\omega_{0}}^{t_1}} \right)^2\left( {{x} - {\omega_{0}}^{t_2}} \right)^2,$$
	by the proof of $   {\bf{Case\; I}}$, we can draw the conclusion that the minimum Hamming distance of $\mathcal{C}_{a}$ is $d_{H}=3$. 
	Similar to $  {\bf{Case\; II}}$, there is no codeword with Hamming weight $4$ and symbol-pair weight $5$. 
	
	If there exsits a codeword with Hamming weight $3$ and symbol-pair weight $5$, the codeword certain cyclic shift must have a form
	$$\left( { \star ,\; \star ,\; {{0_{s}}_{1}} ,\star ,{0_s}_{2}} \right),$$
	where each $ \star $ denotes an element in  ${\mathbb F}_p^{\text{*}}$ and ${0_s}_{1}$, ${0_{s}}_2$ are all-zero vectors  with lengths $s_1$ and $s_2$ respectively. Without loss of generality, suppose that the constant term of $c\left( x \right)$ is 1. We denote that
	$$c\left( x \right) = 1 + {a_1}x + {a_2}{x^t}.$$
	However, $c^{(1)}\left( {\omega_{0}}^{t_1} \right)= c^{(1)}\left( {\omega_{0}}^{t_2} \right) =0$ induces that
	\begin{equation*}  
		\left\{  
		\begin{array}{lr}  
			{a_1} + t{a_2}{\omega_{0}}^{(t-1)t_1} = 0, &\\
			{a_1} + t{a_2}{\omega_{0}}^{(t-1)t_2} = 0. &    
		\end{array}  
		\right.  
	\end{equation*}
	By solving the system, we have ${\omega_{0}}^{(t-1)(t_1-t_2)}=1$, since $t\ne kp$, otherwise ${a_1} = 0$. Together with $\gcd\left({t_1}-{t_2},l \right) =1$, one can immediately get ${l\,\left| \,{t - 1} \right.}$ and ${a_1} + t{a_2} = 0$. Combined with $$c\left( {{\omega_{0}}}^{t_1} \right)= c\left( {\omega_{0}}^{t_2} \right) =0,$$ we have
	\begin{equation*}  
		\left\{  
		\begin{array}{lr}  
			1+	{a_1}{\omega_{0}}^{t_1} + {a_2}{\omega_{0}}^{t_1} = 0, & \\ 
			1+	{a_1}{\omega_{0}}^{t_2} + {a_2}{\omega_{0}}^{t_2} = 0, &   
		\end{array}  
		\right.  
	\end{equation*}
	which implies ${a_1} + {a_2} = 0$. Thus, we have ${p\left| {t - 1} \right.}$, which contradicts the code length $lp$.


As a consequence, we prove that there no exsits a codeword with Hamming weight $3$ and symbol-pair weight $5$. Then, $\mathcal{C}_{a}$ is an MDS $ (lp,6)_{p} $ symbol-pair code, when $({r_1},{r_2} ) =\left(2,2 \right) $ and  $\gcd\left({t_1}-{t_2},l \right) =1$.

\end{proof}	
\begin{corollary}\label{coro C211}
 $\mathcal{C}_{b}$ is an MDS $(lp,6)_{6}$ symbol-pair code.
\end{corollary}
 \begin{proof}
By Lemma \ref{lemma 3.1} and Lemma \ref{lemma 3.3 } the minimum Hamming distance of $\mathcal{C}_{b}$ is 3 and the minimum symbol-pair  distance $d_{p}(\mathcal{C}_{b})\ge 5$, for the generator ploynomial  $ \mathcal{C}_{b} $ is
$$g_{b}\left( x \right) = \left( x - 1 \right)^{2}\left( x - {\omega_{0}}^{t_1} \right)\left(x - {\omega_{0}}^{t_2}\right).$$ 
Using techniques similar to those used in the proof of Proposition \ref{Proposition 3.2}, we see that there are no codewords	of $\mathcal{C}_{b}$ with Hamming weight 4 such that the 4 nonzero terms appear with consecutive coordinates. 
 
 If there exsits a codeword with Hamming weight $3$ and symbol-pair weight $5$, the codeword certain cyclic shift must have a form
 $$\left( { \star ,\; \star ,\; {{0_{s}}_{1}} ,\star ,{0_s}_{2}} \right),$$
 where each $ \star $ denotes an element in  ${\mathbb F}_p^{\text{*}}$ and ${0_s}_{1}$, ${0_{s}}_2$ are all-zero vectors  with lengths $s_1$ and $s_2$ respectively. Without loss of generality, suppose that the constant term of $c\left( x \right)$ is 1. We denote that
 $$c\left( x \right) = 1 + {a_1}x + {a_2}{x^t}.$$
 
	It follows from 
	$c\left( 1 \right) =c^{(1)}\left( 1 \right) = 0$
	that
 \begin{equation*}  
 	\left\{  
 	\begin{array}{lr}  
 		1+	{a_1} + {a_2} = 0, & \\ 
 		  {a_1} + t{a_2} = 0, &
 	\end{array}  
 	\right.  
 \end{equation*}
By solving the system, we have $a_{1}= {-t \over {t-1}}$ and $a_{2}= {1 \over {t-1}}$. Then combined with $c\left( 1 \right) =c\left( {{\omega_{0}}}^{t_1} \right)=0$, we have 
$$t={{\omega_{0}}^{t_1{(t-1)}}}+{{\omega_{0}}^{t_1{(t-2)}}}+\cdots+{{\omega_{0}}^{t_1}}+1.$$
By $a_{1}= {-t \over {t-1}}$, $a_{2}= {1 \over {t-1}}$ and $c\left( {{\omega_{0}}}^{t_1} \right)=0$, one can obtain that 
$$t-1-t{{\omega_{0}}}^{t_1} +{{\omega_{0}}}^{{t_1}t}=0.$$
This implies that 
\[\begin{gathered}
	\;\;\;\;{{\omega_{0}}^{t_1{(t-1)}}}+{{\omega_{0}}^{t_1{(t-2)}}}+\cdots+{{\omega_{0}}^{t_1}}+{{\omega_{0}}}^{{t_1}t}
	+{{\omega_{0}}^{t_1{(t-1)}}}+{{\omega_{0}}^{t_1{(t-2)}}}+\cdots+{{\omega_{0}}^{t_1}}+{{\omega_{0}}}^{{t_1}t}\hfill \\
=2{{\omega_{0}}}^{{t_1}}({{\omega_{0}}^{t_1{(t-1)}}}+{{\omega_{0}}^{t_1{(t-2)}}}+\cdots+{{\omega_{0}}^{t_1}})
	=2t{{\omega_{0}}}^{{t_1}}
	=0,\hfill \\
\end{gathered} \]
which is a contradiction for $c^{(1)}\left( 1 \right) = 0$.

This completes the proof of the Corollary \ref{coro C211}.
 \end{proof}
\begin{remark}
  Let $\omega_{0}$ be the primitive $ l$-th  root of unity in ${{\mathbb F}_p}$ , we can deduce
	\begin{itemize}
	
	\item  $\mathcal{C}_{a}=\left\langle {\left( {{x} - {\omega_{0}}^{t_1}} \right)^2\left( {{x} - {\omega_{0}}^{t_2}} \right)} \right\rangle=\left\langle { (x-1) ^ 2 (x - \omega)} \right\rangle$; 
	\item  $\mathcal{C}_{a}=\left\langle {\left( {{x} - {\omega_{0}}^{t_1}} \right)^2\left( {{x} - {\omega_{0}}^{t_2}} \right)^2} \right\rangle =\left\langle { (x-1) ^ 2 (x - \omega)^2} \right\rangle$;
	\item  $\mathcal{C}_{a}=\left\langle {\left( {{x} - {\omega_{0}}^{t_1}} \right)^3\left( {{x} - {\omega_{0}}^{t_2}} \right)} \right\rangle=\left\langle { (x-1) ^ 3 (x - \omega)} \right\rangle$,
	 	\end{itemize}

\end{remark}
Repeated-root cyclic code $\mathcal{C} =\left\langle { (x-1) ^ 2 (x - \omega)} \right\rangle$ is proposed in Chen \cite{Chen-2017-CIT}. In this paper, $\mathcal{C}=\left\langle { (x-1) ^ 3 (x - \omega)} \right\rangle$ and $\mathcal{C} =\left\langle { (x-1) ^ 2 (x - \omega)^2} \right\rangle$ are are two cases in Theorem \ref{theorem 3.1}.

We use an example to illustrate that the repeated-root cyclic codes of the generator polynomials with the same forms in the above  Corollary \ref{theorem 3.6} are not all MDS symbol-pair codes.
\begin{example} \label{example 3.6}
Let $\mathcal{C}$ and  be a repeated-root cyclic code over ${{\mathbb F}_5}$ and the generator ploynomial of $ \mathcal{C} $  is
$$g\left( x \right) = {\left( {x - 2} \right)^2}\left( {{x} - 3} \right),$$
where $\omega=3$ is a primitive element in ${{\mathbb F}_5}$ and $2=3^3$. Then we have the minimum Hamming distance $d_H=2$ by a magma progarm. Therefore, $\mathcal{C}$ is not an MDS symbol-pair code. 

Similarly, when the generator ploynomial of $ \mathcal{C} $  is one of 	$$g\left( x \right) = {\left( {x - 2} \right)^3}\left( {{x} - 3} \right)$$ and
$$g\left( x \right) = {\left( {x - 2} \right)^2}\left( {{x} - 3} \right)^2,$$
$ \mathcal{C} $ is still not an MDS symbol-pair code, since minimum Hamming distance is $d_H=2$.
\end{example}

Now we present a new class of AMDS symbol-pair codes with the minimum symbol-pair distance 7. 

Let $\mathcal{C}_{1}$ be the cyclic codes over ${{\mathbb F}_p}$. The generator ploynomial of $\mathcal{C}_{1}$ is $$g_{1}\left( x \right) = {\left( {x - 1} \right)^4}\left( {{x} -\omega} \right)\left( {{x} - {\omega}^2} \right),$$  where $\omega$ is a primitive $l$-th  root of unity in ${{\mathbb F}_p}$. 

\begin{theorem}\label{theorem3.9}
	$\mathcal{C}_{1}$ is an AMDS ${\left( {lp,\;7} \right)_p}$ symbol-pair code, if $ l $ odd and $ l \geqslant 3 $.
\end{theorem}
\begin{proof}
	$\mathcal{C}_{1}$ is the cyclic code over ${{\mathbb F}_p}$ generated by $$g_{1}\left( x \right) = {\left( {x - 1} \right)^4}\left( {{x} -\omega} \right)\left( {{x} - {\omega}^2} \right).$$

\noindent	By Lemma \ref{lemma 3.1}, one can derive that $\mathcal{C}_{1}$ is an $\left[ {lp,\;lp - 6,\;4} \right]$ repeated-root cyclic codes code over ${{\mathbb F}_p}$. Lemma \ref{lemma 3.3 } yields that ${d_p} \geqslant 6$, since $\mathcal{C}_{1}$ is not an MDS cyclic code. To prove that $\mathcal{C}_{1}$ is an AMDS symbol-pair code with the minimum symbol-pair distance 7, it is sufficient to verify that there is no a codeword in $\mathcal{ C}_{1}$ with the symbol-pair weight 6. 
	
	If there are codewords in $\mathcal{C}_{1}$ with  Hamming weight 5 and symbol-pair weight 6, then its certain cyclic shift must have the form
	$$ \left( { \star ,\; \star ,\; \star ,\; \star ,\; \star ,{0_s}} \right),$$
	where each $ \star $ denotes an element in  ${\mathbb F}_p^{\text{*}}$ and ${0_s}$ is all-zero vector of length $s$. Without loss of generality, suppose that the constant term of $c\left( x \right)$ is 1. We denote that
	$$c\left( x \right) = 1 + {a_1}x + {a_2}{x^2} + {a_3}{x^3} + {a_4}{x^4} ,$$
	This leads to $\deg\left( {c\left( x \right)} \right) =4 < 6 =\deg\left( {g\left( x \right)} \right)$.
	
	 If $c \in \mathcal{ C}$ has the symbol-pair weight 6 with the Hamming weight 4, then its certain cyclic shift must have the forms
	$$\left( { \star ,\star , \star ,{0_{s_1}}, \star ,{0_{s_2}}} \right)$$ or
	$$ \left( { \star , \star ,{0_{s_1}},\star , \star ,{0_{s_2}}} \right),$$
	where each $ \star $ denotes an element in ${\mathbb F}_p^{\text{*}}$ and ${0_{s_{1}}}$, ${0_{s_{2}}}$ are all-zero vectors  with lengths $s_1$ and $s_2$ respectively.
	
	\item[\textbf{Case I.}] For the case of $$ \left( { \star ,\star , \star ,{0_{s_1}}, \star ,{0_{s_2}}} \right),$$
	without loss of generality, we denote a codeword polynomial
	$$c\left( x \right) = 1 + {a_1}x + {a_2}{x^2} + {a_3}{x^t}.$$
	It follows from $c^{(1)}\left( 1 \right) = c^{(2)}\left( 1 \right) = 0$ that
	\begin{equation*}  
		\left\{  
		\begin{array}{lr}    
			{a_1} + 2{a_2} + t{a_3} = 0,   &\\  
			2{a_2} + t(t-1){a_3} = 0.   & 
		\end{array}  
		\right.  
	\end{equation*}
	By solving the system, we have $ t(t-2){a_3} ={a_1} $. By $c^{(2)}\left( 1 \right) = 0$, we can conclude that 
		\begin{equation*}  
		\left\{  
		\begin{array}{lr} 
			t - 2 \ne kp,k<l,   &\\   
			t - 1 \ne kp,k<l,   &\\  
			t \ne kp,k<l.   & 
		\end{array}  
		\right.  
	\end{equation*}
	 This is a contradiction for  $ c^{(3)}\left( 1 \right)= 0 $ and $a_{3}\in  {\mathbb F}_p^{\text{*}}$.
	
	\item[\textbf{Case II.}]	 For the case of 	$$ \left( { \star , \star ,{0_{s_1}},\star , \star ,{0_{s_2}}} \right),$$
	without loss of generality, we denote
	$$c\left( x \right) = 1 + {a_1}x + {a_2}{x^2} + {a_3}{x^t} + {a_4}{x^{t + 1}}.$$
	It follows from $c\left( 1 \right) =c^{(1)}\left( 1 \right)= 0$ that
	\begin{equation*}  
		\left\{  
		\begin{array}{lr}  
			1 + {a_1} + {a_2} + {a_3} = 0,   &  \\  
			{a_1} + t{a_2} + (t+1){a_3} = 0,  & 
		\end{array}  
		\right.  
	\end{equation*}
	one can derive that  $$ (t-1){a_2} + t{a_3}-1 = 0. \; $$
	By $c^{(2)}\left( 1 \right) = 0$, we have
	$$ t(t-1){a_2} + t(t+1){a_3} = 0. $$
	This leads to $ t({a_3}+1) = 0 $. 
	Therefore, we have $ t=kp,0<k<l $ or $ a_{3}= -1 $.
	
	If $ t=kp,0<k<l $, then 
	\begin{equation*}  
		\left\{  
		\begin{array}{lr}  
			1 + {a_1} + {a_2} + {a_3} = 0, &  \\  
			{a_1} + {a_3} = 0.  & 
		\end{array}  
		\right.  
	\end{equation*}
	This indicates $ {a_1} =- {a_3} $	and $ {a_2}=-1 $. Combined with  $c\left( \omega  \right) = c\left( {\omega ^2} \right) = 0$, we have	
	\begin{equation*}  
		\left\{  
		\begin{array}{lr}  
			{a_1}\omega \left( {{\omega ^t} - 1} \right) = 1 - {\omega ^t}, &  \\  
			{a_1}\omega^2 \left( {{\omega ^{2t}} - 1} \right) = 1 - {\omega ^{2t}}.  & 
		\end{array}  
		\right.  
	\end{equation*}
	By solving the system,we have ${\omega ^{2t}} = 1 $, which contradicts $ l $ odd.
	
	If $ a_{3}= -1 $, by $ c(1)=0 $, we can obtain that $ {a_1} =- {a_2} $. Combined with  $c\left( \omega  \right) = c\left( {\omega ^2} \right) = 0$, we have	
	\begin{equation*}  
		\left\{  
		\begin{array}{lr}  
			{a_1}\omega \left( {{\omega ^{t-1}} - 1} \right) = 1 - {\omega ^{t+1}}, &  \\  
			{a_1}\omega^2 \left( {{\omega ^{2t-2}} - 1} \right) = 1 - {\omega ^{2t+2}}.  & 
		\end{array}  
		\right.  
	\end{equation*}
	
	\noindent Since $\omega$ is a primitive $ l$-th root of unity, then \[\omega \left( {{\omega ^{t - 1}} + 1} \right)\left( {1 - {\omega ^{t + 1}}} \right) = 1 - {\omega ^{2t + 2}}.\]
	This implies that $ \omega ^t = 1 $.
	Thus, $ a_{1}= -1, a_{2}= 1 $.
	By $$ t(t-1){a_2} + t(t+1){a_3} = 0,$$ we have $ 2t=kp $, which contradicts $ \omega ^t = 1 $.

	
	In order to prove that $\mathcal{ C}_{1}$ is an AMDS symbol-pair code, we need to find a codeword with the symbol-pair weight 7. Since 
	$$ c(x)=(x^{p}-1)(x^{p-1}-1)=x^{2p-1}-x^{p}-x^{p-1}+1 $$
	is a codeword of $\mathcal{ C}_{1}$ and $ {\omega _p}\left( {c(x)} \right)=7 $, $\mathcal{C}_{1}$ is an AMDS ${\left( {lp,\;7} \right)_p}$ symbol-pair code. 
\end{proof}

\subsection{MDS and AMDS Symbol-Pair Codes with length $3p$}
In this subsection, we obtain all MDS symbol-pair codes of $d_{p}\le 12$  and all AMDS symbol-pair codes of $d_{p}< 12$ from repeated-root cyclic codes with length $3p$. Furthermore, we discuss all minimum symbol-pair distance of the repeated-root cyclic codes with code length of $3p$, when the degree of generator ploynomials $\deg( g_{(r_{1},r_{2},r_{3})}(x)) \le 10$. 

Let $\mathcal{C}_{(r_{1},r_{2},r_{3})}$ be the repeated-root cyclic code over ${{\mathbb F}_p}$ and the generator ploynomial of $ \mathcal{C}_{(r_{1},r_{2},r_{3})} $ is 
$$g_{(r_{1},r_{2},r_{3})}\left( x \right) = (x - 1)^{r_{1}}(x - \omega)^{r_{2}}(x - {\omega ^2})^{r_{3}}.$$
where $\omega$ is a primitive $3$-th  root of unity in ${{\mathbb F}_p}$ and $r_{i}\le p-1,i=1,2,3$. 
\begin{remark}\label{proposion 3.9}
Let	  $\mathcal{C}_{(r_{1},r_{2},r_{3})}=\left\langle {{{(x - 1)}^{r_{1}}}{{(x - \omega)}^{r_{2}}}{{(x - {\omega ^{2}})}^{r_{3}}}} \right\rangle$ have the same minimum symbol-pair distance, if the exponents of the three factors of the generator polynomial can be swapped.
\end{remark}
\begin{proof}
	We first prove that such repeated-root cyclic codes $$\widetilde{\mathcal{C}} = \left\langle {{{(x - \omega^i)}^{r_{1}}}{{(x - \omega^{i+1})}^{r_{2}}}{{(x - {\omega ^{i+2}})}^{r_{3}}}} \right\rangle $$ are the same codes for $i=0,1,2$. 
	
	Without loss of generality, suppose that $$\widetilde g_1\left( x \right) = (x - 1)^{r_{1}}(x - \omega)^{r_{2}}(x - \omega ^2)^{r_{3}} ,$$ 
	 $${\widetilde g_2}\left( x \right) = (x - 1)^{r_{3}}(x - \omega)^{r_{1}}(x - \omega^2)^{r_{2}} $$ 
	 and 
	 $${\widetilde g_3}\left( x \right) = (x - 1)^{r_{2}}(x - \omega)^{r_{3}}(x - {\omega ^2})^{r_{1}}.$$ 
	 We denote that $\widetilde g_1\left( x \right)$, $\widetilde g_2\left( x \right)$ and $\widetilde g_3\left( x \right)$ represent the generator polynomials of $\widetilde {\mathcal{C}_1}$, $\widetilde {\mathcal{C}_2}$ and $\widetilde {\mathcal{C}_3}$, respectively.
	
	Let $y = \frac{x}{{{\omega ^2}}},z = \frac{x}{\omega }$, for the generator polynomial $\widetilde g_1\left( x \right)$ of $\widetilde{{\mathcal{C}}_1}$, we can deduce the following results by deforming it.
	\[\begin{gathered}
		{\widetilde g_1}(x) = {(x - 1)^{{r_1}}}{(x - \omega )^{{r_2}}}{(x - {\omega ^2})^{{r_3}}} \hfill \\
		\;\;\;\;\;\;\;\;\;\;\,= {\omega ^{2\left( {{r_1} + {r_2} + {r_3}} \right)}}{(\frac{x}{{{\omega ^2}}} - \frac{1}{{{\omega ^2}}})^{{r_1}}}{(\frac{x}{{{\omega ^2}}} - \frac{\omega }{{{\omega ^2}}})^{{r_2}}}{(\frac{x}{{{\omega ^2}}} - \frac{{{\omega ^2}}}{{{\omega ^2}}})^{{r_3}}} \hfill \\
		\;\;\;\;\;\;\;\;\;\;\,= {\omega ^{2\left( {{r_1} + {r_2} + {r_3}} \right)}}{(y - \omega )^{{r_1}}}{(y - {\omega ^2})^{{r_2}}}{(y - 1)^{{r_3}}} \hfill \\
		\;\;\;\;\;\;\;\;\;\;\,={\omega ^{2\left( {{r_1} + {r_2} + {r_3}} \right)}}{\widetilde g_{2}}(y) \hfill \\
		\;\;\;\;\;\;\;\;\;\;\,= {\omega ^{{r_1} + {r_2} + {r_3}}}{(\frac{x}{\omega } - \frac{1}{\omega })^{{r_1}}}{(\frac{x}{{{\omega ^2}}} - \frac{\omega }{\omega })^{{r_2}}}{(\frac{x}{\omega } - \frac{{{\omega ^2}}}{\omega })^{{r_3}}} \hfill \\
		\;\;\;\;\;\;\;\;\;\;\,= {\omega ^{{r_1} + {r_2} + {r_3}}}{(z - {\omega ^2})^{{r_1}}}{(z - 1)^{{r_2}}}{(z - \omega )^{{r_3}}} \hfill \\ 
		\;\;\;\;\;\;\;\;\;\;\,= {\omega ^{{r_1} + {r_2} + {r_3}}}{\widetilde g_3}(z). \hfill \\
	\end{gathered} \]
	\noindent Thus, repeated-root cyclic codes $\widetilde{\mathcal{C}} $ are the same codes for $i=0,1,2$.
	
	Next, since  $\omega$ and $\omega^2$ are primitive 3-th root of unity in ${\mathbb F}_{p}$, we have   $\widetilde {\mathcal{C}_1} $ and  $\widetilde {\mathcal{C}_4} $  have the same minimum symbol-pair distance, where the generator polynomial of $\widetilde {\mathcal{C}_4}$ is
	 $$\widetilde g_{4}(x)= (x - 1)^{r_{1}}(x - \omega)^{r_{3}}(x - {\omega ^2})^{r_{2}}.$$
Therefore, all $\widetilde{ \mathcal{C}}$ have the same minimum symbol-pair distance, when  the exponents of the three factors exchanged with each other.	
\end{proof}
The above Remark \ref{proposion 3.9} shows that the exponential positions of the three factors $x-1$, $x-\omega$ and $x-\omega^2$ of the generator polynomial of 
$ \mathcal{C}_{(r_{1},r_{2},r_{3})} $ have the same minimum symbol-pair distance. Without loss of generality, suppose that $$p-1 \ge r_{1}\ge r_{2}\ge r_{3}\ge 0$$ in the next part of this subsection. Then we have the following theorem.
\begin{theorem}\label{theorem of 3p MDS}
	$ \mathcal{C}_{(r_{1},r_{2},r_{3})} $ is an MDS symbol-pair codes over ${\mathbb F}_{p}$, if one of the following two conditions is true
	\begin{enumerate}
		\item $r_{1}\le 5$, $0\le r_{2}-r_{3}\le 1$ and $ {r_{1}} = r_{2}+r_{3}$,
		\item $r_{1}< 5$, $0\le r_{2}-r_{3}\le 1$ and $ {r_{1}} = r_{2}+r_{3}+1$.
	\end{enumerate}
\end{theorem}

Researchers in \cite{Chen-2017-CIT} and \cite{Ma-2022-DCC} given some proofs of Theorem \ref{theorem of 3p MDS}, and the Theorem \ref{theorem 3.1} in the previous paper also includes some proofs. Here we only need to prove that ${\mathcal{C}_{(4,2,1)}} $ is an MDS symbol-pair code.

\begin{Proposition}\label{Proposition C_{421}}
${\mathcal{C}_{(4,2,1)}} $ is an MDS ${\left( {3p,\;9} \right)_p}$ symbol-pair code.	
\end{Proposition} 
\begin{proof}
Since  $\mathcal{C}_{(4,2,1)}= \left\langle {{{(x - 1)}^4}{{(x - \omega )}^2}{{(x - {\omega ^2})}}} \right\rangle ,$ for any codeword $c \in \mathcal{C}_{(4,2,1)}$, we have $$c\left( 1 \right) = c\left( \omega  \right) = c\left( {  \omega^2 } \right) = c^{(1)}\left( 1 \right)= c^{(1)}\left( \omega \right) = c^{(2)}\left( 1 \right) =c^{(3)}\left( 1 \right) = 0.$$
By Lemma \ref{lemma 3.1}, $\mathcal{C}_{(4,2,1)}$ is a $\left[ {3p,\;3p - 7,\;5} \right]$ cyclic code over ${{\mathbb F}_p}$. Since $\mathcal{C}_{(4,2,1)}$ is a subcode of Lemma \ref{lemma 3.6 }, we have ${d_p} \geqslant 8$.

To prove that $\mathcal{C}_{(4,2,1)}$ is an MDS ${\left( {3p,\;9} \right)_p}$ symbol-pair code, it suffices to verify that there does not exist codeword in $\mathcal{C}_{(4,2,1)}$ with symbol-pair weight 8. Then we have three cases to discuss.
	
	\item[\textbf{Case I.}] If there are codewords with Hamming weight $5$ and symbol-pair weight $8$, then its certain
	cyclic shift must be one of the following forms
	$$ \left( { \star , \star ,{0_{s_1}}, \star ,\star,{0_{s_2}}, \star ,{0_{s_3}}} \right)$$ or
	$$ \left( { \star ,\star , \star ,{0_{s_1}}, \star ,{0_{s_2}}, \star ,{0_{s_3}}} \right),$$
	where each $ \star $ denotes an element in ${\mathbb F}_p^{\text{*}}$ and  ${0_{s_1}}$, ${0_{s_2}}$,${0_{s_3}}$ are all-zero vectors  with lengths $s_1$, $s_2$ and $s_3$ respectively.
	
\item 
$  {\bf{Subcase\; 1.1.}}$ For the case of $$ \left( { \star , \star ,{0_{s_1}}, \star ,\star,{0_{s_2}}, \star ,{0_{s_3}}} \right),$$ without loss of generality, suppose that the constant term of $c\left( x \right)$ is 1. We denote that

$$c\left( x \right) = 1 + {a_1}x + {a_2}{x^l} + {a_3}{x^{l + 1}} + {a_4}{x^t}.$$

When $t \equiv 0({\bmod~3}) $ and $l \equiv 0({\bmod~3}) $, it follows from $c\left( 1 \right) = c\left( { \omega} \right) =c\left( { \omega^2} \right) = 0$ that
\begin{equation*}  
	\left\{  
	\begin{array}{lr}  
		1+{a_1} + {a_2} + {a_3} + {a_4} = 0, &  \\  
		1+{a_1}\omega + {a_2} + {a_3}\omega + {a_4} = 0, &\\  
		1+{a_1}\omega^2 + {a_2} + {a_3}\omega^2 + {a_4}  = 0. &    
	\end{array}  
	\right.  
\end{equation*}
By solving the system, we have	 $ a_1=-a_3 $. However, $c^{(1)}\left( 1 \right) =c^{(1)}\left( \omega \right) =0$ induces that
\begin{equation*}  
	\left\{  
	\begin{array}{lr}   
		{a_1} + t{a_2} +\left( t+1\right) {a_3} + l{a_4}  = 0, &\\
		{a_1} + t{a_2}\omega^2 +\left( t+1\right) {a_3} + l{a_4}\omega^2  = 0. &    
	\end{array}  
	\right.  
\end{equation*}
Together with $  a_1=-a_3 $, one can immediately get
can get $t(1-\omega^2)a_3=0$, which is impossible, since $t \equiv 0\left( {\bmod~ 3} \right)$ and the code length is $3p$. 

When $l \equiv i({\bmod~3}) $ and $t \equiv j({\bmod~3}),\, i,j=0,1,2 $,
values in all $i$  and $j$ of ${\bf{Subcase\; 1.1}}$ are shown in the following Table \ref{table-2}.

\begin{table}[h]\label{table2}
		\centering
	\begin{threeparttable}[b]
		\caption{Summary of  ${\bf{Subcase\; 1.1}}$}\label{table-2}
			\begin{tabular}{@{}lllll@{}}			
				\toprule
				
				$i$ & $j$ & Conditions &\makecell{Results} & Contradictory \\
				\midrule
				0 & 0 &$\left[\kern-0.15em\left[ 1 
				\right]\kern-0.15em\right],\left[\kern-0.15em\left[ 2 
				\right]\kern-0.15em\right]$ & \makecell{$t(1-\omega^2)a_3=0$} & $t\le  3p-2$ \\ 
				\midrule
				0 & 1 &$\left[\kern-0.15em\left[ 1 
				\right]\kern-0.15em\right],\left[\kern-0.15em\left[ 2 
				\right]\kern-0.15em\right],\left[\kern-0.15em\left[ 3 
				\right]\kern-0.15em\right]$& \makecell{$\omega^2-1=0 $  } & $\omega^3=1$ \\ 
				\midrule
				0 & 2 &$\left[\kern-0.15em\left[ 1 
				\right]\kern-0.15em\right],\left[\kern-0.15em\left[ 3 
				\right]\kern-0.15em\right]$ & \makecell{$ a_1+a_3+a_4=0,$\\$a_1+a_3-a_4=0 $ }& $a_4\in {\mathbb F}_p^{\text{*}}$\\
				\midrule
				1 & 0 &$\left[\kern-0.15em\left[ 1 
				\right]\kern-0.15em\right],\left[\kern-0.15em\left[ 3 
				\right]\kern-0.15em\right]$ & \makecell{$ a_1+a_2+a_3=0,$\\$a_1+a_2-a_3=0 $ }& $a_3\in {\mathbb F}_p^{\text{*}}$\\
				\midrule
				1 & 1 & $\left[\kern-0.15em\left[ 1 
				\right]\kern-0.15em\right],\left[\kern-0.15em\left[ 3 
				\right]\kern-0.15em\right]$ & \makecell{$ 3=0 $ }& $p\ne 3$\\ 
				\midrule
				1 & 2 & $\left[\kern-0.15em\left[ 1 
				\right]\kern-0.15em\right],\left[\kern-0.15em\left[ 3 
				\right]\kern-0.15em\right]$ & \makecell{$ 3=0 $ }& $p\ne 3$\\ 
				\midrule
				2 & 0 & $\left[\kern-0.15em\left[ 1 
				\right]\kern-0.15em\right],\left[\kern-0.15em\left[ 3 
				\right]\kern-0.15em\right]$ & \makecell{$ a_1-a_2=0 $,\\$a_1+a_2=0$ }& $a_1,a_2 \in {\mathbb F}_p^{\text{*}}$.\\ 
				\midrule
				2 & 1 & $\left[\kern-0.15em\left[ 1 
				\right]\kern-0.15em\right],\left[\kern-0.15em\left[ 3 
				\right]\kern-0.15em\right]$ & \makecell{$a_1+a_2+a_4=0,$\\$a_1-a_2+a_4=0 $ }& $a_2 \in {\mathbb F}_p^{\text{*}}$\\ 
				\midrule
				2 & 2 &$\left[\kern-0.15em\left[ 1 
				\right]\kern-0.15em\right],\left[\kern-0.15em\left[ 3 
				\right]\kern-0.15em\right]$ &  \makecell{$a_1+a_2+a_4=0,$\\$a_1-a_2-a_4=0 $} &  $a_1 \in {\mathbb F}_p^{\text{*}}$ \\ 
					\midrule
			\end{tabular}
		\begin{tablenotes}
			\item * {\footnotesize Conditions $\left[\kern-0.15em\left[ 1 \right]\kern-0.15em\right]$, $\left[\kern-0.15em\left[ 2 \right]\kern-0.15em\right]$  and $\left[\kern-0.15em\left[ 3 \right]\kern-0.15em\right]$ represent $c\left( 1 \right) = c\left( { \omega} \right) =c\left( { \omega^2} \right) = 0$, $ c^{(1)}\left( 1 \right)= c^{(1)}\left( \omega \right) =0$ and  $c\left( 1 \right) + c\left( { \omega} \right) +c\left( { \omega^2} \right) = 0$, respectively.}
		\end{tablenotes}
	\end{threeparttable}
\end{table}

	\item 
	$  {\bf{Subcase\; 1.2.}}$ For the subcase of $$ \left( { \star , \star , \star ,{0_{s_1}},\star,{0_{s_2}}, \star ,{0_{s_3}}} \right),$$ without loss of generality, suppose that the constant term of $c\left( x \right)$ is 1. We denote that
	
	$$c\left( x \right) = 1 + {a_1}x + {a_2}{x^2} + {a_3}{x^l} + {a_4}{x^t}.$$
	
Suppose that $l \equiv i({\bmod~3}) $ and $t \equiv j({\bmod~3}),\, i,j=0,1,2 $, similar to ${\bf{Subcase\; 1.1}}$, we summarize all $i$ and $j$ of ${\bf{Subcase\; 1.2}}$ in the following Table \ref{table-3}.

	\begin{table}\label{table3}
	\centering
	\begin{threeparttable}[b]
		\caption{Summary of ${\bf{Subcase\; 1.2}}$}\label{table-3}
		\begin{tabular}{lllll}
			\hline
			$i$ & $j$ & Conditions  & \makecell{Results}& Contradictory \\
			\midrule
			0 & 0 &$\left[\kern-0.15em\left[ 1 
			\right]\kern-0.15em\right],\left[\kern-0.15em\left[ 2 
			\right]\kern-0.15em\right]$ & \makecell{$ a_1-a_2=0 $,\\$a_1+a_2=0$ }& $a_1,a_2 \in {\mathbb F}_p^{\text{*}}$\\  
			\midrule
			0 & 1 & $\left[\kern-0.15em\left[ 1 \right]\kern-0.15em\right],\left[\kern-0.15em\left[ 2 \right]\kern-0.15em\right]$ & \makecell{$ a_1-a_2+a_4=0 $,\\$a_1+a_2+a_4=0$ } & $a_2 \in {\mathbb F}_p^{\text{*}}$\\
			\midrule  
			0 & 2 & $\left[\kern-0.15em\left[ 1 \right]\kern-0.15em\right],\left[\kern-0.15em\left[ 2 \right]\kern-0.15em\right]$ & \makecell{$ a_1+a_2+a_4=0 $,\\$a_1-a_2-a_4=0$ } & $a_1 \in {\mathbb F}_p^{\text{*}}$\\
			\midrule
			1 & 1 & $\left[\kern-0.15em\left[ 1 \right]\kern-0.15em\right],\left[\kern-0.15em\left[ 2 \right]\kern-0.15em\right]$ & $\makecell{3=0}$ & $p\ne 3$ \\
			\midrule 
			1 & 2 & $\left[\kern-0.15em\left[1\right]\kern-0.15em\right],\left[\kern-0.15em\left[ 2 \right]\kern-0.15em\right]$ & $\makecell{3=0}$ & $p\ne 3$ \\
			\midrule
			2 & 2 & $\left[\kern-0.15em\left[1\right]\kern-0.15em\right],\left[\kern-0.15em\left[ 2 \right]\kern-0.15em\right]$ &\makecell{$3=0$} & $p\ne 3 $\\
			\hline
		\end{tabular}
		\begin{tablenotes}
			\item * {\footnotesize Conditions $\left[\kern-0.15em\left[ 1 
				\right]\kern-0.15em\right]$ and $\left[\kern-0.15em\left[ 2 \right]\kern-0.15em\right]$  represent $c\left( 1 \right) = c\left( { \omega} \right) =c\left( { \omega^2} \right) = 0$ and $c\left( 1 \right) + c\left( { \omega} \right) +c\left( { \omega^2} \right) = 0$,  respectively.}
		\end{tablenotes}
	\end{threeparttable}
\end{table}

	\item[\textbf{Case II.}] If there are codewords with Hamming weight $6$ and symbol-pair weight $8$, then its certain
	cyclic shift must be one of the following forms
	$$ \left(  \star , \star ,\star ,\star ,\star ,{0_{s_1}}, \star,{0_{s_2}} \right),$$ 
	$$ \left(  \star ,\star ,\star , \star ,{0_{s_1}}, \star ,\star ,{0_{s_2}}\right)\;\;$$ or 
	$$ \left(  \star ,\star ,\star , {0_{s_1}},\star ,\star , \star ,{0_{s_2}}\right),$$
	where each $ \star $ denotes an element in ${\mathbb F}_p^{\text{*}}$ and  ${0_{s_1}}$, ${0_{s_2}}$ are all-zero vectors   with lengths $s_1$ and $s_2$ respectively.
	
\item 	${\bf{Subcase\; 2.1.}}$ For the subcase of $$ \left( { \star , \star , \star,\star , \star ,{0_{s_1}},\star,{0_{s_2}}} \right),$$ without loss of generality, suppose that the constant term of $c\left( x \right)$ is 1. We denote 
	
	$$c\left( x \right) = 1 + {a_1}x + {a_2}{x^2} + {a_3}{x^{ 3} }+{a_4}{x^4} + {a_5}{x^{ t} } .$$

	 When $t \equiv 0\left( {\bmod~3} \right)$, it can be derived from
	 	\begin{equation*}  
	 	\left\{  
	 	\begin{array}{lr}  
	 		c\left( 1 \right) = c\left( { \omega} \right) =c\left( { \omega^2} \right) = 0, &  \\  
	 		c\left( 1 \right) + c\left( { \omega} \right) +c\left( { \omega^2} \right) = 0, &
	 	\end{array}  
	 	\right.  
	 \end{equation*}
	  that
		\begin{equation*}  
			\left\{  
			\begin{array}{lr}  
				{a_1}+{a_2}+{a_4}=0, &  \\  
				{a_1}-{a_2}+{a_4}=0, &
			\end{array}  
			\right.  
		\end{equation*}
	 which is impossible, since $a_{2} \in {\mathbb F}_p^{\text{*}}$ and $p$ is an odd prime.				
		
	 When $t \equiv 1\left( {\bmod~3} \right)$, with arguments similar to $t \equiv 0\left( {\bmod~3} \right)$, a contradiction can be obtained
		from		 	
		\begin{equation*}  
			\left\{  
			\begin{array}{lr}  
				c\left( 1 \right) = c\left( { \omega} \right) =c\left( { \omega^2} \right) = 0, &  \\  
				c\left( 1 \right) + c\left( { \omega} \right) +c\left( { \omega^2} \right) = 0. &
			\end{array}  
			\right.  
		\end{equation*}	
		
	When $t \equiv 2\left( {\bmod~3} \right)$, it follows from $c\left( 1 \right) = c\left( { \omega} \right) =c\left( { \omega^2} \right) = 0$ that
		\begin{equation*}  
			\left\{  
			\begin{array}{lr}  
				1+{a_1}+{a_2}+{a_3}+{a_4}+{a_5}=0, &  \\  
				1+{a_1}\omega+{a_2}\omega^2+{a_3}+{a_4}\omega+{a_5}\omega^2=0, &\\  
				1+{a_1}\omega^2+{a_2}\omega+{a_3}+{a_4}\omega^2+ {a_5}\omega=0. &    
			\end{array}  
			\right.  
		\end{equation*}
		By solving the system, we have $1+a_3=0,a_1+a_4=0,a_2+a_5=0$. Then $c^{(1)}\left( 1 \right)= c^{(1)}\left( \omega \right) =0$  indicates
		\begin{equation*}  
			\left\{  
			\begin{array}{lr}  
				{a_1} + 2{a_2} +3 {a_3} + 4{a_4}+ t{a_5}  = 0, &\\
				{a_1} + 2{a_2}\omega +3 {a_3}\omega^2 + 4{a_4}+ t{a_5} \omega^2 = 0, &    
			\end{array}  
			\right.  
		\end{equation*}
		which means that ${a_5} = {{3{a_3}\omega^2 } \over {t - 2}}= {{3{a_4}\omega } \over {t - 2}}$ (since $t \equiv 2\left( {\bmod~3} \right)$, then $p$ is not a divisor of $t-2$, otherwise $t-2 \ge 3p$).
		
	\noindent	By $ c^{(2)}\left( 1 \right) =0$, we have $t = 3+2\omega$. Together with ${a_5} = {{3{a_3}\omega^2 } \over {t - 2}}= {{3{a_4}\omega } \over {t - 2}}$ and $c^{(3)}\left( 1 \right) = 0$, one can derive that
		$$ 6+24\omega+3t(t-1)\omega^2=0.$$ 
		Then we have $$2+8\omega+(3+2\omega)(2+2\omega)\omega^2=0.$$ Combining with $\omega^2=-1-\omega$, we can obtain $3\omega^2=0$. This is impossible.

\item 	${\bf{Subcase\; 2.2.}}$ For the subcase of $$ \left( { \star , \star , \star,\star ,{0_{s_1}}, \star ,\star,{0_{s_2}}} \right),$$ without loss of generality, suppose that the constant term of $c\left( x \right)$ is 1. We denote that
	
	$$c\left( x \right) = 1 + {a_1}x + {a_2}{x^2} + {a_3}{x^{ 3} }+{a_4}{x^t} + {a_5}{x^{ t+1} } .$$

		 When $t \equiv 0\left( {\bmod~3} \right)$  and $t \equiv 2\left( {\bmod~3} \right)$, with arguments similar to the previous $t \equiv 0\left( {\bmod~3} \right)$ of ${\bf{Subcase\; 2.1}}$,	a contradiction can be derived from
		$c\left( 1 \right) = c\left( { \omega} \right) =c\left( { \omega^2} \right) = 0$ again.	
		
	 When $t \equiv 1\left( {\bmod~3} \right)$, with arguments similar to $t \equiv 2\left( {\bmod~3} \right)$ of ${\bf{Subcase\; 2.1}}$, $1+a_3=0,a_1+a_4=0,$ and $a_2+a_5=0$ can be obtained from $c\left( 1 \right) = c\left( { \omega} \right) =c\left( { \omega^2} \right) = 0$.
		
\noindent Then ${a_5} ={a_4}\omega= {{3{a_3}\omega^2 } \over {t - 1}}$ can be derived from $c^{(1)}\left( 1 \right)= c^{(1)}\left( \omega \right) =0$. $c^{(2)}\left( 1 \right)=0$ means $t=-\omega$. Finally, combined with $c^{(3)}\left( 1 \right)=0$, we have $\omega^2-\omega = 0$, a contradiction again.			

\item 	${\bf{Subcase\; 2.3.}}$ For the subcase of $$ \left( { \star , \star , \star,{0_{s_1}},\star , \star ,\star,{0_{s_2}}} \right),$$ without loss of generality, suppose that the constant term of $c\left( x \right)$ is 1. We denote that
	
	$$c\left( x \right) = 1 + {a_1}x + {a_2}{x^2} + {a_3}{x^{ t} }+{a_4}{x^{t+1}} + {a_5}{x^{ t+2} } .$$

When $t \equiv 0\left( {\bmod~3} \right)$, it follows from $c\left( 1 \right) = c\left( { \omega} \right) =c\left( { \omega^2} \right) = 0$ that
	\begin{equation*}  
		\left\{  
		\begin{array}{lr}  
			1+{a_1}+{a_2}+{a_3}+{a_4}+{a_5}=0, &  \\  
			1+{a_1}\omega+{a_2}\omega^2+{a_3}+{a_4}\omega+{a_5}\omega^2=0, &\\  
			1+{a_1}\omega^2+{a_2}\omega+{a_3}+{a_4}\omega^2+ {a_5}\omega=0. &    
		\end{array}  
		\right.  
	\end{equation*}
	By solving the system, we have  $1+a_3=0,a_1+a_4=0,a_2+a_5=0$.	
	Then $c^{(1)}\left( 1 \right)= c^{(1)}\left( \omega \right) =0$  indicates
	\begin{equation*}  
		\left\{  
		\begin{array}{lr}  
			{a_1} + 2{a_2} +t {a_3} + \left( t+1\right) {a_4}+ \left( t+2\right) {a_5}  = 0, &\\
			{a_1} + 2{a_2}\omega +t {a_3}\omega^2 + \left( t+1\right) {a_4}+ \left( t+2\right){a_5} \omega = 0, &    
		\end{array}  
		\right.  
	\end{equation*}
	which means that $  {a_5} ={a_4}\omega=a_3\omega^2$. Then $c^{(2)}\left( 1 \right)= 0$ implies that
	$$2{a_2} +t\left( t-1\right)  {a_3} +t \left( t+1\right) {a_4}+ \left( t+1\right) \left( t+2\right) {a_5}  = 0,$$
	which implies $t(3\omega^2+\omega-1=0) $. Since  $t \equiv 0\left( {\bmod 3} \right), \omega=-\omega^2-1$ and the code length $3p$, we have $2(\omega^2-1)=0$, a contradiction.				
	
	When $t \equiv 1\left( {\bmod~3} \right)$, with arguments similar to the $t \equiv 0\left( {\bmod~3} \right)$ of $\bf{Subcase\; 2.3}$, by 
		\begin{equation*}  
		\left\{  
		\begin{array}{lr}  
			c\left( 1 \right) = c\left( { \omega} \right) =c\left( { \omega^2} \right) = 0, &\\
			c^{(1)}\left( 1 \right)= c^{(1)}\left( \omega \right) =0, &    
		\end{array}  
		\right.  
	\end{equation*}
 we have ${a_4} = {a_3}\omega= {{{a_5}\omega^2{\left( t +2\right) } } \over {t-1}}$. Together with $ c^{(2)}\left( 1 \right)  = 0$, $\omega^2-\omega =0$ can be derived, which is impossible.

When $t \equiv 2\left( {\bmod~3} \right)$, similarly, we can derive ${a_5}={a_4}\omega =  {{{a_3}\omega^2{\left( t - 2\right) } } \over {t+1}}$ from  
		\begin{equation*}  
	\left\{  
	\begin{array}{lr}  
		c\left( 1 \right) = c\left( { \omega} \right) =c\left( { \omega^2} \right) = 0, &\\
		c^{(1)}\left( 1 \right)= c^{(1)}\left( \omega \right) =0. &    
	\end{array}  
	\right.  
\end{equation*}
Then, together with $ c^{(2)}\left( 1 \right)  = 0$, we have $t =p+1$. It follows from $c^{(1)}\left( 1 \right)= c^{(1)}\left( \omega \right) =0$ that
	 	\begin{equation*}  
	 	\left\{  
	 	\begin{array}{lr}  
	 		(t-1) {a_3} + \left( t+1\right) {a_4}+ \left( t+1\right) {a_5}  = 0, &\\
	 		(t-1) {a_3}\omega + \left( t+1\right) {a_4}\omega^2+ \left( t+1\right){a_5}  = 0. &    
	 	\end{array}  
	 	\right.  
	 \end{equation*}
	 Combined with $t =p+1$,  we have $\omega^2=1$, which contradicts that $\omega$ is primitive 3-th root of unity in ${\mathbb F}_{p}$.
	

	\item[\textbf{Case III.}] If there are codewords in $\mathcal{C}$ with Hamming weight 7 and symbol-pair weight 8, then its certain cyclic shift must have the form $$ \left( { \star , \star , \star , \star ,\star , \star , \star ,{0_s}} \right),$$ where each $ \star $ denotes an element in  ${\mathbb F}_p^{\text{*}}$ and ${0_s}$ is all-zero vector of length $s$. Without loss of generality, suppose that the constant term of $c\left( x \right)$ is 1. We denote
	
	$$c\left( x \right) = 1 + {a_1}x + {a_2}{x^2} + {a_3}{x^3} + {a_4}{x^4} + {a_5}{x^5}+ {a_6}{x^6}.$$
	This leads to $\deg\left( {c\left( x \right)} \right) = 6 < 7 = \deg\left( {g\left( x \right)} \right)$.


As a result,  $\mathcal{C}_{(4,2,1)}$ is an MDS ${\left( {3p,{\text{\;}}9} \right)_p}$ symbol-pair code.	
\end{proof}

Based on Reference \cite{Chen-2017-CIT}, \cite{Zhao-2019-Doc}, \cite{Ma-2022-DCC} , Remark \ref{proposion 3.9}, Proposition \ref{Proposition C_{421}} to Proposition \ref{proposition of C_{3rr}} and Corollary \ref{coro dp=11} in this paper, all known MDS symbol-pair codes with $n-k\le 10$ from repeated-root cyclic codes  ${\mathcal{C}_{(r_{1},r_{2},r_{3})}}, $ which are listed in the following Table \ref{table-4}.
\begin{table}[h]\label{table4}
	\centering
	\begin{threeparttable}[b]
		\caption{All MDS symbol-pair codes with length $3p$ for $d_{p}\le 12$}\label{table-4}
		\begin{tabular}{@{}lllll@{}}			
			\toprule
			
			$r_1$ & $r_2$ & $r_3$ & \makecell{$(n-k,d_{p})_{p}$} & Reference or Proposition \\
			\midrule	
			0 & 2 & 0 & \makecell{$(2,4)_{p}$} & Trivially($r_1= 0$)\\
			\midrule  
			2 & 1 & 0 & \makecell{$(3,5)_{p}$} & Reference \cite{Chen-2017-CIT} \\
			\midrule
			2 & 1 & 1 & \makecell{$(4,6)_{p}$} & Reference \cite{Chen-2017-CIT} \\
			\midrule			
			3 & 1 & 0 & \makecell{$(4,6)_{p}$} & Reference \cite{Zhao-2019-Doc} \\ 
			\midrule					
			3 & 1 & 1 & \makecell{$(5,7)_{p}$} & Reference \cite{Chen-2017-CIT}\\ 			
			\midrule
			3 & 2 & 1 & \makecell{$(6,8)_{p}$} &Reference \cite{Chen-2017-CIT}\\
			\midrule
			4 & 2 & 2 & \makecell{$(8,10)_{p}$}&Reference \cite{Ma-2022-DCC} \\ 
			\midrule
			5 & 3 & 2 & \makecell{$(10,12)_{p}$} &Reference \cite{Ma-2022-DCC} \\
			\midrule
			2 & 2 & 0 & \makecell{$(4,6)_{p}$} & Proposition \ref{Proposition 3.4} \\
			\midrule
			4 & 2 & 1 & \makecell{$(7,9)_{p}$} &Proposition \ref{Proposition C_{421}}\\
			\midrule
			
		\end{tabular}	
		\begin{tablenotes}
			\item * {\footnotesize These MDS symbol-pair codes are constructed by repeated-root cyclic codes.}
		\end{tablenotes}
	\end{threeparttable}
\end{table}

Next, we will explain that there is no MDS symbol pair code except in Table \ref{table-4}, when the degree of generator polynomials $g_{(r_{1},r_{2},r_{3})}(x)$ does not exceed 10.

Similarly, all AMDS symbol-pair codes with $d_{p}<12$ can be deduced from the following propositions. We list all AMDS symbol-pair codes in the following Table \ref{table5}.	
	\begin{table}\label{table-5}
		\centering
		\begin{threeparttable}[b]
			\caption{All AMDS symbol-pair codes with length $3p$ for $d_{p}< 12$}\label{table5}
			\begin{tabular}{lllll}
				\hline
				$r_1$ & $r_2$ & $r_3$ & \makecell{$(n-k,d_{p})_{p}$} & Reference or Proposition \\
				\midrule
				4 & 3 & 2 & \makecell{$(9,10)_{p}$} &Reference \cite{Ma-2022-DCC} \\
				\midrule 
				0 & 3 & 0 & \makecell{$(3,4)_{p}$} & Proposition \ref{proposition of C_{rrr}}\\
				\midrule
				2 & 2 & 1 & \makecell{$(5,6)_{p}$} & Proposition \ref{proposition of C_{rrr}}\\ 
				\midrule					
				3 & 2 & 0 & \makecell{$(5,6)_{p}$} &Proposition \ref{proposition of C_{rrr}}\\ 
				\midrule										
				3 & 2 & 2 & \makecell{$(7,8)_{p}$} &Proposition \ref{proposition of C_{3rr}} \\ 
				\midrule
				3 & 3 & 1 & \makecell{$(7,8)_{p}$} &Proposition \ref{proposition of C_{3rr}} \\ 
				\midrule
				4 & 1 & 0 & \makecell{$(5,6)_{p}$} & Proposition \ref{proposition of C_{rrr}} \\
				\midrule 
				4 & 1 & 1 & \makecell{$(6,7)_{p}$} &Proposition \ref{proposition of C_{3rr}}\\
				\midrule
				4 & 3 & 1 & \makecell{$(8,9)_{p}$} &Proposition \ref{proposition of C_{3rr}}\\
				\midrule
				5 & 2 & 1 & \makecell{$(8,9)_{p}$} &Proposition \ref{proposition of C_{3rr}} \\
				\midrule					
				5 & 2 & 2 &\makecell{$(9,10)_{p}$}&Proposition  \ref{proposition of C_{3rr}}\\
				\hline
			\end{tabular}
			\begin{tablenotes}
			\item * {\footnotesize These AMDS symbol-pair codes are constructed by repeated-root cyclic codes.}
			\end{tablenotes}
		\end{threeparttable}
	\end{table}

In what follows, let's determine the minimum symbol-pair distance for $\mathcal{C}_{(r_{1},r_{2},r_{3})}$ in the previous paper by using the following propositions.

\begin{Proposition}\label{proposition of C_{rrr}}
	\item 
\begin{enumerate}
\item The minimum distance of $\mathcal{C}_{(0,r_{2},0)}$ is 4, when $2\le r_{2} \le p-1$.
\item The minimum distance of $\mathcal{C}_{(2,1,0)}$ is 5.
\item The minimum distance of $\mathcal{C}_{(r_{1},r_{2},0)}$ is 6, when   $r_{1}+r_{2} \ge 4, r_{2}\ge 1$.
\item The minimum distance of $\mathcal{C}_{(2,{r_{2},r_{3})}}$ is 6, when  $2\le r_{2}+r_{3}\le 4$.
\end{enumerate}	
\end{Proposition}
\begin{proof}

 For $\mathcal{C}_{(0,r_{2},0)}=\left\langle {{(x - \omega )}^{r_{2}}} \right\rangle$, Lemma \ref{lemma 3.1} shows $d_{H}(\mathcal{C}_{(0,r_{2},0)})=2$, then one can obtain $d_{p}(\mathcal{C}_{(0,r_{2},0)})=4$ by Lemma \ref{lemma 3.3 }.

Corollary \ref{theorem 3.6} proves the minimum distance of $\mathcal{C}_{(2,1,0)}$ is $d_{p}=5$.
	
Since $\mathcal{C}_{(r_{1},r_{2},0)}$ is a subcode of Theorem \ref{theorem 3.6} and Lemma \ref{lemma 3.1} means $d_{H}(\mathcal{C}_{(r_{1},r_{2},0)})=3$. Then we have $d_{p}(\mathcal{C}_{(r_{1},r_{2},0)})=6$.

Since $\mathcal{C}_{(2,{r_{2},r_{3}})}$ is a subcode of  $\mathcal{C}_{(2,1,1)}$ and $\mathcal{C}_{(2,1,1)}$ is an MDS symbol-pair code with the mimnmum symbol-pair distance 6, when   $2\le r_{1}+r_{2} \le 4$ and $r_{1},r_{2} \in {\mathbb F}_p^{\text{*}}$. Then we can deduce   $d_{p}(\mathcal{C}_{(2,{r_{2},r_{3}})})=6$.
\end{proof}
We can obtain all MDS and AMDS symbol-pair codes with length $3p$ with a minimum symbol-pair distance of 4 to 6 using Proposition \ref{proposition of C_{rrr}}. Then, we check for MDS and AMDS symbol-pair codes with a symbol pair distance of 7 to 12.

\begin{Proposition}\label{proposition of C_{3rr}}
		\item 
	\begin{enumerate}
		\item The minimum distance of $\mathcal{C}_{(r_{1},1,1)}$ is 7, when   $3\le r_{1}\le p-1$.
		\item The minimum distance of $\mathcal{C}_{(3,r_{2},r_{3})}$ is 8, when   $3\le r_{2}+r_{3}\le6, r_{3}\ge 1$.
		\item The minimum distance of $\mathcal{C}_{(r_{1},r_{2},1)}$ is 9, when  $2\le r_{2}\le r_{1}\le p-1, r_{1}\ge 4$.
		\item 	The minimum distance of $\mathcal{C}_{(r_{1},r_{2},r_{3})}$ is 10, when   $ r_{1}$, $ r_{2}$ and $ r_{3}$ meet any of the following two conditions
		\begin{itemize}
			\item $4\le r_{1}\le p-1$ and  $ r_{2}= r_{3}=2$,
			\item $r_{1}= 4$ and  $4\le r_{2}+ r_{3}\le 8, r_{3}\ge 1$.
		\end{itemize}
		\end{enumerate}
\end{Proposition}
\begin{proof}
	
	In reference \cite{Chen-2017-CIT},  $\mathcal{C}_{(2,1,1)}$, $\mathcal{C}_{(3,1,1)}$ and $\mathcal{C}_{(3,2,1)}$ are MDS symbol-pair codes with symbol-pair distances of 6, 7 and 8, respectively.  Reference \cite{Ma-2022-DCC} proved $\mathcal{C}_{(4,2,2)}$ and $\mathcal{C}_{(5,3,2)}$ are MDS symbol-pair codes with symbol-pair distances of 10 and 12, respectively.	
\begin{enumerate}

	\item 	Since $\mathcal{C}_{(r_{1},1,1)}$ is a subcode of $\mathcal{C}_{(3,1,1)}$, we have  $d_{p}(\mathcal{C}_{(r_{1},1,1)})\ge 7$. 
	Then the minimum symbol-pair distance $d_{p}(\mathcal{C}_{(r_{1},1,1)})= 7$ can be obtained by $\omega_p(c(x))=7$, where $$c(x)=1-x-x^p+x^{2p+1}$$ is a codeword of  $\mathcal{C}_{(r_{1},1,1)}$.
	
	\item  Reference \cite{Chen-2017-CIT} proved that $\mathcal{C}_{(3,2,1)}$ is an MDS $(3p,8)_{p}$ symbol-pair code. Since $3\le r_{2}+r_{3}\le6$, we have $\mathcal{C}_{(3,r_{2},r_{3})}$ is subcode of $\mathcal{C}_{(3,2,1)}$ and $d_{p}(\mathcal{C}_{(3,r_{2},r_{3})})\ge 8$.
	
	By Lemma \ref{lemma 3.1}, we can deduce that the minimum Hamming distance of $\mathcal{C}_{(3,r_{2},r_{3})}$ is 4, which implies  $d_{p}(\mathcal{C}_{(3,r_{2},r_{3})})\le 8$. 
	
	Therefore, the minimum distance of $\mathcal{C}_{(3,r_{2},r_{3})}$ is $d_{p}=8$, when   $3\le r_{2}+r_{3}\le6, r_{3}\ge 1$.
	
	\item With arguments similar as the proof of case, since $\mathcal{C}_{(4,2,1)}$ is an MDS symbol-pair code with the minimum symbol-pair distance 9, we can deduce that the minimum symbol-pair distance of $\mathcal{C}_{(r_{1},r_{2},1)}$ is $d_{p}(\mathcal{C}_{(r_{1},r_{2},1)})\ge 9$. Next, we prove that there are symbol-pair codewords with symbol-pair weight $d_{p}=9$ in $\mathcal{C}_{(r_{1},r_{2},1)}$.
	
	For the codeword 
	$$c(x)=(x-1)^p(x-\omega)^p(x-\omega^2)= {x^{2p + 1}} - {\omega ^2}{x^{2p}} + {\omega ^2}{x^{p + 1}} - \omega {x^p} + \omega x - 1,$$
	it is easy to verify that $c(x)$ is a codeword polynomial of $\mathcal{C}_{(r_{1},r_{2},1)}$ with the symbol-pair weight 9.
	Thus, we have $d_{p}(\mathcal{C}_{(r_{1},r_{2},1)})=9$.
	
	\item For $4\le r_{1}\le p-1$ and  $ r_{2}= r_{3}=2$, with arguments similar as the proof of case, since $\mathcal{C}_{(4,2,2)}$ is an MDS symbol-pair code with $d_{p}=10$ and $\mathcal{C}_{(r_{1},2,2)}$ is a subcode of $\mathcal{C}_{(4,2,2)}$, we have $d_{p}(\mathcal{C}_{(r_{1},2,2)})\ge 10$. 
	Note that the codeword polynomial
	$$c(x)=1-x^2+2x^{p+1}+x^{p+2}-x^{2p}-2x^{2p+1}$$
	is a codeword of $\mathcal{C}_{(r_{1},2,2)}$ and $\omega_p(c(x))=10$. Therefore, we derive the minimum symbol-pair distance of $\mathcal{C}_{(r_{1},2,2)}$ is 10.
	
	For $r_{1}= 4$ and  $4\le r_{2}+ r_{3}\le 8$, since $\mathcal{C}_{(4,r_{2},r_{3})}$ is a subcode of $\mathcal{C}_{(4,2,2)}$, we have $d_{p}(\mathcal{C}_{(4,r_{2},r_{3})})\ge 10$.	However, Lemma \ref{lemma 3.1} shows that $d_{H}(\mathcal{C}_{(4,r_{2},r_{3})})=5$, which implies that $d_{p}(\mathcal{C}_{(4,r_{2},r_{3})})\le 10$. Thus, we can deduce $d_{p}(\mathcal{C}_{(4,r_{2},r_{3})})= 10$.
\end{enumerate}		
\end{proof}
 According to Proposition \ref{proposition of C_{3rr}}, we can obtain all MDS and AMDS symbol-pair codes with length $3p$ with a minimum symbol-pair distance of 7 to 10. We can also deduce that $\mathcal{C}_{(r_{1},r_{2},r_{3})}$ is an MDS symbol-pair code with $d_{p}=12$, iff $(r_{1},r_{2},r_{3})=(5,3,2)$. Furthermore, there dose not exist codeword with a minimum symbol pair distance of 11, when the degree of the generator polynomials  $\deg (g_{(r_{1},r_{2},r_{3})}(x))$ are 9 and 10. The following corollary can be drawn.
\begin{corollary} \label{coro dp=11}
The repeated-root cyclic code $\mathcal{C}_{(r_{1},r_{2},r_{3})}$ must not be the MDS and the AMDS symbol-pair code with a minimum symbol-pair distance 11. 	
\end{corollary}
Based on the above Proposition \ref{Proposition C_{421}} to Proposition \ref{proposition of C_{3rr}}, Corollary \ref{coro dp=11} and Remark  \ref{proposion 3.9}, we complete the proof of Theorem \ref{theorem of 3p MDS}. In addition, we also explain that the MDS symbol-pair codes in Table \ref{table-4} and the AMDS symbol-pair codes in Table \ref{table5} are all cases that meet the requirements.

Proposion \ref{proposition of C_{3rr}} shows that the condition does not satisfy Theorem \ref{theorem of 3p MDS}, when $r\ge 5$ and $r_{1} = r_{2} + r_{3} + 1$. Next, we use an example to illustrate that the conditions of  Theorem \ref{theorem of 3p MDS}  are also no longer applicable, when $r>5$ and $r_{1} = r_{2} + r_{3} $.
\begin{example}
	Let $\mathcal{C}$ and  be a repeated-root cyclic code over ${{\mathbb F}_7}$ and the generator ploynomial of $ \mathcal{C} $  is
	$$g\left( x \right) = (x-1)^6(x - 2)^3(x - 4 )^3,$$
	where $\omega=2$ is a 3-th primitive element in ${{\mathbb F}_7}$ and $2^2=4$. 
	
	Then we have the minimum Hamming distance $d_H=7$ by a magma progarm.
	Reference \cite{Ma-2022-DCC} shows the symbol-pair distance $d_p\ge 12$. The magma program also shows that both vectors
	$$\textbf{a}=[0\; 0\; 1\; 0\; 0\; 0\; 0\; 0\; 0\; 6\; 6\; 3\; 3\; 3\; 4\; 4\; 4\; 5\; 1\; 1\; 1]$$
	and
	$$\textbf{b}=[0\; 0\; 0\; 0\; 0\; 0\; 1\; 0\; 0\; 1\; 1\; 1\; 5\; 4\; 4\; 4\; 3\; 3\; 3\; 6\; 6]$$
	are in $\mathcal{C}$. Let $\textbf{c}=\textbf{a+b}$, we have
	$$\textbf{c}=[0\; 0\; 1\; 0\; 0\; 0\; 1\; 0\; 0\; 0\; 0\; 4\; 1\; 0\; 1\; 1\; 0\; 1\; 4\; 0\; 0],$$
	which is also in $\mathcal{C}$. We can easily deduce $\omega_p(\textbf{c})=13$.
	Therefore, $\mathcal{C}$ is not an MDS symbol-pair code.	
\end{example}

\section{Conclusion}
In this paper, employing repeated-root cyclic codes, some new classes of MDS and AMDS symbol-pair codes over ${\mathbb F}_p$
with lengths $ lp $  and $3p$ are provided. We give some more general generator polynomials about MDS $(lp,5)_{p}$ and $(lp,6)_{p}$ symbol-pair codes. We also present a class of AMDS $(lp,7)_{p}$ symbol-pair codes. For length $3p$, we provide all MDS symbol-pair codes with $d_{p}\le 12$ and also provide all AMDS symbol-pair codes with $d_{p}< 12$.

\end{document}